\documentclass[submission,copyright]{eptcs}

\providecommand{\publicationstatus}{To appear in EPTCS.}

\usepackage{makeidx,fancybox}
\usepackage{amsmath,hhline,amsthm}
\usepackage{amssymb,array,amscd}
\usepackage[all,2cell]{xy}

\title{Turchin's Relation for Call-by-Name Computations:\\ A Formal Approach}
\author{Antonina Nepeivoda
\institute{Program Systems Institute of Russian Academy of Sciences\thanks{The reported study was partially supported by RFBR, research project No. 14-07-00133, and Russian Academy of Sciences, research project No. AAAA-A16-116021760039-0.}\\
  Pereslavl-Zalessky, Russia\\}
\email{a\_nevod@mail.ru}
}
\def\titlerunning{Turchin's Relation for Call-by-Name Computations}
\def\authorrunning{A. Nepeivoda}

\usepackage{color}            

\newcommand*{\ovl}[1]{\ensuremath{\overline{\raisebox{0pt}[1.3\height]{#1}}}}

\newtheorem{Theorem}{Theorem}
\newtheorem{Prop}{Proposition}

\newtheorem{Definition}{Definition}

\newtheorem{Example}{Example}
\def\logand{\mathrel{\&}}

\def\logimpl{\mathrel{\Rightarrow}}

\def\child{\mathop{\textrm{child}}}
\def\nil{\mathop{\textrm{[Stop]}}}
\def\pop{\mathop{\textrm{[Pop]}}}

\def\App{\mathop{\textrm{App}}}
\def\Copy{\mathop{\textrm{Copy}}}
\def\Lift{\mathop{\textrm{Ins}}}
\def\Del{\mathop{\textrm{Del}}}
\def\deriv{\mathop{\textrm{deriv}}}
\def\inv{\mathop{\textrm{inv}}}
\def\lw{\mathop{\mathsf{LW}(\Upsilon,\mathbf{S})}}

\begin{document}
\maketitle

\begin{abstract}
Supercompilation is a program transformation technique that was first described by V.\,F.~Turchin in the 1970s. In supercompilation, Turchin's relation as a similarity relation on call-stack configurations is used both for call-by-value and call-by-name semantics to terminate unfolding of the program being transformed. In this paper, we give a formal grammar model of call-by-name stack behaviour. We classify the model in terms of the Chomsky hierarchy and then formally prove that Turchin's relation can terminate all computations generated by the model. 
\end{abstract}


\section{Introduction}

Supercompilation is a program transformation method based on fold/unfold operations \cite{Tur86,Jones,Hamilton}. Given a program and its parameterized input configuration, a supercompiler partially unfolds the computation tree of the program on the input configuration and then tries to fold the tree back into a graph, which presents the residual program. In the general case, the computation tree may be infinite. Thus, the following question appears: when is it reasonable to stop the unfolding in order to avoid going into an infinite loop?

One of the ways to solve this problem is based on ``configuration similarity'' relations. If a path in the tree contains two configurations, the latter of which resembles the former, that may be a sign that the path represents an unfolded loop. Thus, when a supercompiler finds two such configurations, it terminates unfolding of the path where they appear. In general, the loop recognizing problem is undecidable. Hence, we must make a choice: either to take a risk of an infinite unfolding trying to find all the finite paths, or to take a risk of terminating the finite paths too early guaranteeing termination of all the infinite paths. In most supercompilers, the second option is preferred\footnote{In supercompiler SCP4 \cite{NemytykhBook}, we can choose one of the two options.} \cite{Leuschel,Gluck95,Hamilton}.

Now we recall an important relation property used for termination. 

\begin{Definition}\label{definition:awo}
Given a set $T$ of terms and a set $S$ of sequences of the terms from $T$, relation \(R\subset T\times T\) is called \emph{a well binary relation with respect to set $S$}, if every sequence \(\{\Phi_n\}\in S\) such that \(\forall i,j (i<j\logimpl (\Phi_i, \Phi_j)\notin R)\) is finite \cite{Leuschel}. 

\end{Definition}

So, a well binary relation is ``a well quasi-order without the order'' (i.\,e., it is not necessarily transitive).

Any relation guaranteeing termination of the unfolding of a computation tree must be a well binary relation with respect to the set of the traces generated in the tree. The relation most widely used for this aim, the homeomorphic embedding \cite{Leuschel,Albert,Gluck95}, is well binary with respect to arbitrary term sequences \cite{Kruskal}. Some other relations used for termination in program transformations\footnote{Among them is the relation used in supercompiler SCP4 \cite{NemytykhBook} and the relation used in higher-order supercompiler HOSC \cite{Kluch}.} are not well binary with respect to arbitrary term sequences. However, they are well binary with respect to term sequences that can be generated on any computation path\footnote{This property of not being well binary for arbitrary sequences makes it harder to prove the well-binariness of the relations, because the ``minimal bad sequence'' reasoning, which is used, e.g., in the classical proof of Kruskal's theorem in \cite{NW65}, does not work.}. 

This paper studies Turchin's relation, well-binariness of which also can be proved only with respect to computation paths that appear during unfolding. That relation on call-stack configurations was the first well binary relation used for trace termination \cite{Tur86} (1986). Although Turchin's relation is a useful tool that helps to solve both termination and generalization problems \cite{Turchin88} (also see Section~\ref{Section:Preliminaries} of this paper), the proof of its well-binariness given by V.~Turchin in \cite{Turchin88} was presented in a semi-formal way. For the call-by-value semantics, the formal proof of this property of Turchin's relation is given in \cite{ANNTur}. The formalization is based on the prefix rewriting grammars model. But as far as we know, for the call-by-name semantics, the relation was never formally studied. This paper tries to cover this gap.  

Our contributions are the following:

\begin{enumerate}
\item We introduce a notion of a multi-layer prefix grammar. Elements of traces generated by such a grammar are call-stack configurations on computation paths in the call-by-name semantics. We show that the class of grammars is stronger than the class of context-free grammars. 
\item We prove a strengthened version of Turchin's theorem on well-binariness of Turchin's relation. Namely, we prove that every infinite computation path modelled by a trace of a multi-layer prefix grammar contains an infinite chain w.r.t. Turchin's relation. As a consequence, one can use Turchin's relation in composition with the homeomorphic embedding relation (or any other relation, which is well-binary on the arbitrary sequences of terms) without the loss of well-binariness. \footnote{The idea behind the composition is to make Turchin's relation responsible for approximating call-stack behavior of the program, while the second relation defined on the configurations in the infinite Turchin's chain has to take into account static properties of the terms in the infinite chain above. Such a composition may allow a supercompiler to construct more accurate generalizations (see Example~\ref{Example:log2progcomposite}).}.\end{enumerate}

The paper is organized as follows. In Section~\ref{Section:Preliminaries}, we informally introduce Turchin's relation for call-stack configurations. In Section~\ref{section:MLPG} we define a class of grammars that model call stack behaviour for call-by-name computations and describe its computational power. In Section~\ref{section:MCSBMLG} we show how such grammars can be used for modelling the call stack behaviour of programs in a simple functional language. Finally, in Section~\ref{section:TurRelMLPG} we refine the~definition of Turchin's relation for the new class of the grammars.

\subsection{Presentation Language}
\label{Subsection:PresLang}

In this subsection, we informally describe the syntax and semantics of the simple functional language $\mathbb{L}$ used below for demonstrating the modelling method.

The language $\mathbb{L}$ is based on the call-by-name semantics. The names of the variables in $\mathbb{L}$ are the~words starting with the~letter ${x}$. Let $\mathcal{E}(\mathbb{L})$ denote the set of expressions in $\mathbb{L}$, 
then:

\begin{enumerate}
\item if $t$ is a variable or a constant (a null-ary constructor), then $t\in \mathcal{E}(\mathbb{L})$,

\item if $t_i\in \mathcal{E}(\mathbb{L})$ ($1\leq i\leq n$), and $C$ is a constructor of the arity $n$, then $C(t_1,\dots,t_n)\in \mathcal{E}(\mathbb{L})$,

\item if $t_i\in \mathcal{E}(\mathbb{L})$ ($1\leq i\leq n$), and $f$ is a function name of the arity $n$, then $f(t_1,\dots,t_n)\in \mathcal{E}(\mathbb{L})$.
\end{enumerate} 

For the sake of brevity, we use natural numbers $n\in\mathbb{N}$ for denoting terms of the form $S(S(\dots S(Z)\dots))$ (the unary Peano numbers), where the constant $Z$ stands for $0$. Thus, we denote the increment constructor $S(w)$ by $w+1$.

A~definition of a~function ${f(x_1,...,x_n)}$ in $\mathbb{L}$ is a~sequence of sentences of the~form 

$${f(T_1,...,T_n)=P};$$

Here ${T_i}$ is either a variable or expression $C(t_1,\dots, t_n)$, where $C$ is a constructor and $t_i$ are variables. The expression on the right-hand side of the definition, ${P}$, is an arbitrary expression in $\mathcal{E}(\mathbb{L})$ containing only the variables occurring in ${f(T_1,...,T_n)}$. For every left-hand side of the definition ${f(T_1,...,T_n)}$, no variable can appear in ${f(T_1,...,T_n)}$ more than once.

The $\mathbb{L}$-program sentences are rewriting rules. The rewriting rules in the programs are ordered from top to bottom and they should be matched in this order (as in the Haskell and Refal \cite{TurRefal5} languages).

A simple program in $\mathbb{L}$ is given in Example~\ref{Example:log2prog}.

\section{Turchin's Relation}
\label{Section:Preliminaries}

\subsection{Turchin's Relation for Call-Stack Configurations}

The relation we refer to was described in early 1970's by Valentin Turchin in his seminal works on supercompilation. Turchin's relation considers every call stack as a list of the function names, starting from the top. The idea is the following. If a loop occurs along a computation path, then some prefix of the list of function names must repeat itself. Then the common suffix (maybe empty) of the lists representing the function names points to computations after the loop. Turchin's relation ignores arguments and considers only call-stack configurations before successful attempts to evaluate a function call. Thus, the actions of the call stack restructuring are ignored by Turchin's relation.

\begin{Example} 
\label{Example:log2prog}

The following program defines function $f(x)=[\log_2(x)]+1$ for natural $x\geq 1$. A part of its computation tree starting from the input point $f(h(x))$ is shown in Figure~\ref{fig:1treelog2prog}. In the tree, only stack configurations \emph{before evaluating the program rules} are present. For example, the first call-stack configuration ${out:=f(h(x))}$ is omitted, because the attempt to compute the call $f$ with the argument $h(x)$ does not lead to any substitutions. It causes only the restructuring of the stack.

\footnotesize {$$\begin{array}{|l|l|}\hline \textrm{A program computing }[\log_2(x)]+1\\\hline
{f(0)=0;}\\ {f(x+1)=f(g(x+1))+1;}\\\\
{g(0)=0;}\\ {g(x+1)=h(x);}\\\\
{h(0)=0;}\\ {h(x+1)=g(x)+1;}\\\hline
\end{array}$$}

\begin{figure}[t]
\footnotesize {\[
\UseAllTwocells
\xymatrix @-4mm {
&&&*+[F-:<5pt>]{\begin{array}{ll} \textrm{Term: }&f(h(x))\\\textrm{Stack: }&{z_0:=h(x)}\\&{out:=f(z_0)}\end{array}}\ar[dll]_<(.35)
{x=0}\ar[d]^<(.25){x=x_1+1}\\
&*++[F-:<5pt>]{{f(0)}}\ar[d] && *+[F-:<5pt>]{\begin{array}{ll} \textrm{Term: }&{f(g(x_1)+1)}\\\textrm{Stack: }&{out:=f(g(x_1)+1)}\end{array}}\ar[d]\\
&*++[F-:<5pt>]{{0}}  && *+[F-:<5pt>]{\begin{array}{ll} \textrm{Term: }&{f(g(g(x_1)+1))+1}\\\textrm{Stack: }&{z_1:= g(g(x_1)+1),}\\ &{out:=f(z_1)+1}\end{array}}\ar[d]\\
&*++[F-:<5pt>]{{f(h(0))+1}}\ar[dl] && *+[F-:<5pt>]{\begin{array}{ll} \textrm{Term: }&{f(h(g(x_1)))+1}\\\textrm{Stack: }&{z_2:= g(x_1),}\\&{z_1:= h(z_2),}\\ &{out:=f(z_1)+1}
\end{array}}\ar[ll]_<(.3){x_1=0}\ar[d]^{x_1=x_2+1}\\
*++[F-:<5pt>]{{f(0)+1}}\ar[d]&&& *+[F-:<5pt>]{\begin{array}{ll} \textrm{Term: }&{f(h(h(x_2)))+1}\\\textrm{Stack: }&{z_2:= h(x_2),}\\&{z_1:= h(z_2),}\\ &{out:=f(z_1)+1}
\end{array}}\ar[dll]_<(.35){x_2=0}\ar[d]^<(.2){x_2=x_3+1}\\
*++[F-:<5pt>]{{1}}&*++[F-:<5pt>]{{f(h(0))+1}} && *+[F-:<5pt>]{\begin{array}{ll} \textrm{Term: }&{f(h(g(x_3)+1))+1}\\\textrm{Stack: }&{z_1:= h(g(x_3)+1),}\\ &{out:=f(z_1)+1}
\end{array}}
}\] }
\caption{A fragment of the computation tree for the program of Example~\ref{Example:log2prog}}
\label{fig:1treelog2prog}
\end{figure}
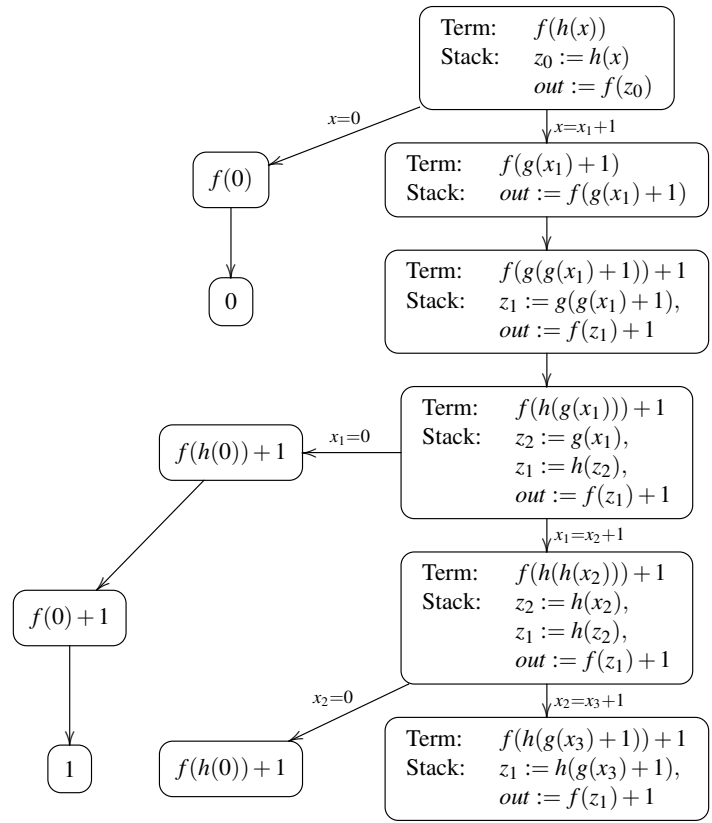

\noindent \normalsize The two configurations ${f(h(g(x_3)+1))+1}$ and ${f(h(x))}$ have a similar stack structure. Namely, both their stacks contain the two calls: the top is the call of function $h$ and the bottom is the call of function $f$. 
\end{Example}

Turchin's relation fits well for finding similarities in the call-stack configurations as the one in Example~\ref{Example:log2prog}. Namely, it checks whether the two call stacks $\Delta_1$ and $\Delta_2$ on the path can be split into parts $\mathit{[Top]}$, $\mathit{[Middle]}$, and $\mathit{[Context]}$ such that $\Delta_1=\mathit{[Top]}\mathit{[Context]}$, $\Delta_2=\mathit{[Top]}\mathit{[Middle]}\mathit{[Context]}$ and the part $\mathit{[Context]}$ is never changed on the path segment starting at $\Delta_1$ and ending at $\Delta_2$ (as shown in Figure~\ref{fig:turpic}). The passive part of data is ignored; only function names in the call stacks are considered. Thus, a call stack is treated by Turchin's relation as \emph{a word consisting of function names in the call stack} (maybe with some annotation). 

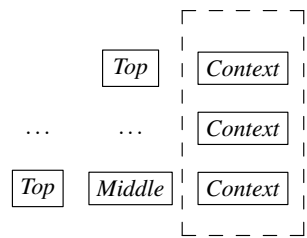
\begin{figure}[h]
\footnotesize
\UseAllTwocells
\UseTips
$$
\xymatrix @-5mm {& & & \\
& *+[F]{{\,\mathit{Top}\,}}& *+[F]{{\mathit{Context}}}
\\
\dots & \dots&*+[F]{{\mathit{Context}}}
\\ 
*+[F]{{\mathit{Top}}} & *+[F]{{\mathit{Middle}}};& *+[F]{{\mathit{Context}}}\save "1,3"."4,3"*++[F--]\frm{}\restore
\\
}$$
\caption{Turchin's relation for call-stack configurations}
\label{fig:turpic}
\end{figure}

Looking back at Example~\ref{Example:log2prog}, we can infer that the call-stack configurations for ${f(h(g(x_3)+1))+1}$ and ${f(h(x))}$ satisfy Turchin's relation. But the first two configurations that satisfy Turchin's relation in the tree are ${f(g(g(x_1)+1))+1}$ and ${f(h(g(x_1)))+1}$. The innermost call of the first configuration,  ${g(x_1)}$, does not appear in the call stack, because the call is in the passive part of the configuration. Hence, the call stacks of the two configurations are modelled by the words ${gf}$ and ${ghf}$. Moreover, the call $f$ is unchanged on the path segment starting at ${f(g(g(x_1)+1))+1}$ and ending at ${f(h(g(x_1)))+1}$. So we can assign $\mathit{[Top]}=g$, $\mathit{[Middle]}=h$, $\mathit{[Context]}=f$. 

It is worth noting that Turchin's relation additionally provides the following generalization strategy (a description of generalization can be found in, e.\,g., \cite{Gluck95}). If the two stacks are of the form $\mathit{[Top]}\mathit{[Context]}$ and $\mathit{[Top]}\mathit{[Middle]}\mathit{[Context]}$, then the former term can be decomposed to $\mathbf{let}\textrm{ }v_0=\mathit{[Top]}\textrm{ }\mathbf{in}\textrm{ }\mathit{[Context]}$. After the decomposition, the parts $\mathit{[Top]}$ and $\mathit{[Context]}$ are developed by supercompilation separately. Thus, the output format of the part $\mathit{[Top]}$ is not seen any more by the part $\mathit{[Context]}$.

\begin{Example} 
\label{Example:log2proggen}

Figure~\ref{fig:treegenTur} shows how the computation tree of Example~\ref{Example:log2prog} can be generalized with the use of Turchin's relation (Figure~\ref{fig:treegenTur}). For the sake of brevity, the call stacks are represented by the words consisting of only the function names (from the top to the bottom)\footnote{We recall that $w+1$ is a notation for $S(w)$, where $S$ is a constructor (see Subsection~\ref{Subsection:PresLang}).}.

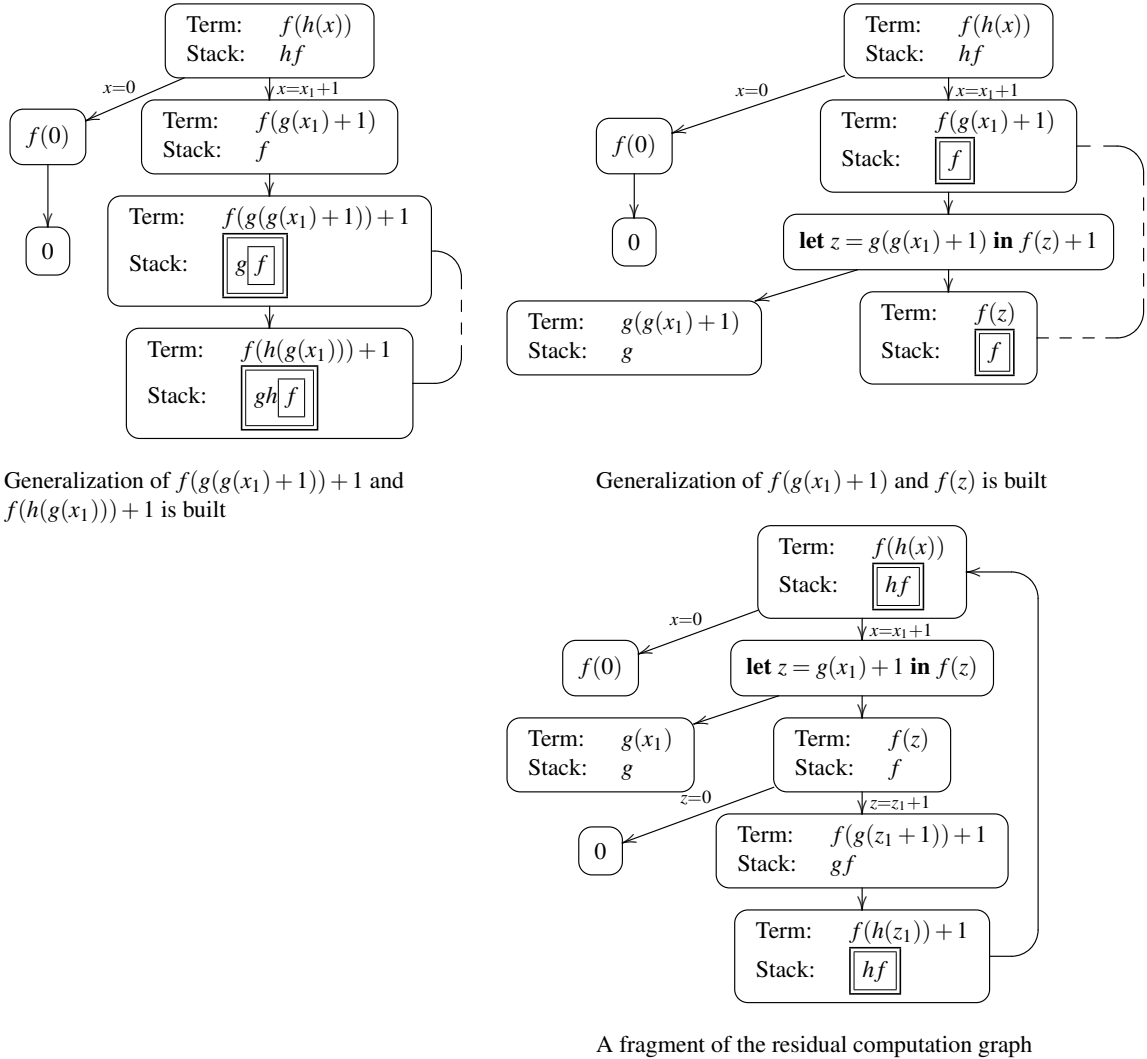
\begin{figure}[t]

\footnotesize
\begin{tabular}{lll}
 $$
\xymatrix @-5.5mm {
&*+[F-:<5pt>]{\begin{array}{ll} \textrm{Term: }&f(h(x))\\\textrm{Stack: }&{hf}\end{array}}\ar[dl]_<(.4){x=0}\ar[d]^<(.2){x=x_1+1}\\
*++[F-:<5pt>]{f(0)}\ar[d] & *+[F-:<5pt>]{\begin{array}{ll} \textrm{Term: }&{f(g(x_1)+1)}\\\textrm{Stack: }&{f}\end{array}}\ar[d]\\
*++[F-:<5pt>]{0}  & *+[F-:<5pt>]{\begin{array}{ll} \textrm{Term: }&{f(g(g(x_1)+1))+1}\\\textrm{Stack: }&{\doublebox{$g$\framebox[2\width]{$f$}}} \end{array}}\ar[d]&\\
 & *+[F-:<5pt>]{\begin{array}{ll} \textrm{Term: }&{f(h(g(x_1)))+1}\\\textrm{Stack: }&{\doublebox{${gh}$\framebox[2\width]{$f$}}}\end{array}}\ar@{--}`r[ru]`[ul][u]
}
$$ 
&
$$
\xymatrix @-5.5mm {
&*+[F-:<5pt>]{\begin{array}{ll} \textrm{Term: }&f(h(x))\\\textrm{Stack: }&{hf}\end{array}}\ar[dl]_<(.4){x=0}\ar[d]^<(.2){x=x_1+1}\\
*++[F-:<5pt>]{f(0)}\ar[d] & *+[F-:<5pt>]{\begin{array}{ll} \textrm{Term: }&{f(g(x_1)+1)}\\\textrm{Stack: }&{\doublebox{$f$}}\end{array}}\ar[d]\\
*++[F-:<5pt>]{0}  & *++[F-:<5pt>]{\mathbf{let}\textrm{ }{z=g(g(x_1)+1)}\textrm{ }\mathbf{in}\textrm{ }{f(z)+1}}\ar[dl]\ar[d]&\\
*+[F-:<5pt>]{\begin{array}{ll} \textrm{Term: }&g(g(x_1)+1)\\\textrm{Stack: }&{g}\end{array}} & *+[F-:<5pt>]{\begin{array}{ll} \textrm{Term: }&f(z)\\\textrm{Stack: }&{\doublebox{$f$}}\end{array}}\ar@{--}`r[ru]`[uul][uu]
}
$$ \\ \\
Generalization of ${f(g(g(x_1)+1))+1}$ and & \qquad \qquad Generalization of ${f(g(x_1)+1)}$ and
${f(z)}$ is built\\${f(h(g(x_1)))+1}$ is built \\
&
 $$
\xymatrix @-5.5mm {
&*+[F-:<5pt>]{\begin{array}{ll} \textrm{Term: }&f(h(x))\\\textrm{Stack: }&{\doublebox{$hf$}}\end{array}}\ar[dl]_<(.4){x=0}\ar[d]^<(.25){x=x_1+1}\\
*++[F-:<5pt>]{f(0)} & *++[F-:<5pt>]{\mathbf{let}\textrm{ }{z=g(x_1)+1}\textrm{ }\mathbf{in}\textrm{ }{f(z)}}\ar[dl]\ar[d]\\
*+[F-:<5pt>]{\begin{array}{ll} \textrm{Term: }&g(x_1)\\\textrm{Stack: }&{g}\end{array}} & *+[F-:<5pt>]{\begin{array}{ll} \textrm{Term: }&f(z)\\\textrm{Stack: }&{f}\end{array}}\ar[dl]_<(.4){z=0}\ar[d]^{z=z_1+1}\\
*++[F-:<5pt>]{0} & *+[F-:<5pt>]{\begin{array}{ll} \textrm{Term: }&{f(g(z_1+1))+1}\\\textrm{Stack: }&{gf}\end{array}}\ar[d]&\\
& *+[F-:<5pt>]{\begin{array}{ll} \textrm{Term: }&{f(h(z_1))+1}\\\textrm{Stack: }&{\doublebox{$hf$}}\end{array}}\ar@{->}`r[ru]`[uuuul][uuuu]
}
$$\\ \\
& \qquad \qquad A fragment of the residual computation graph 
\end{tabular}
\caption{Generalization by Turchin's relation}
\label{fig:treegenTur}
\end{figure}

\noindent Figure~\ref{fig:progTurMSG} shows the residual programs generated by the final computation graph constructed by generalization w.r.t. Turchin's relation and by the graph generalized with the use of the homeomorphic embedding relation (as described in \cite{Gluck95}). Turchin's relation constructs a shorter residual program with better efficiency. In Example~\ref{Example:log2progcomposite} below, we will see that the composition of the two relations allows us to produce even better generalization. In order to use the composition, we must first prove that the composition is well-binary. The proof is given in Section~\ref{section:TurRelMLPG}.

\begin{figure}[h]
\footnotesize{
$$
\begin{array}{lll}
\begin{array}{|l|l|}\hline \textrm{Residual program on the base} \\ \textrm{of Turchin's relation (according to Figure~\ref{fig:treegenTur})}\\\hline
{f_1(0)=0;}\\ {f_1(x+1)=f_2(g_1(x)+1);}\\\\
{f_2(0)=0;}\\ {f_2(x+1)=f_1(x)+1;}\\\\
{g_1(0)=0;}\\ {g_1(1)=0;}\\
{g_1(x+1+1)=g_1(x)+1;}\\\hline
\end{array}
&\qquad &
\begin{array}{|l|l|}\hline \textrm{Residual program on the base} \\ \textrm{of the homeomorphic embedding}\\\hline
{f_1(0)=0;}\\ {f_1(x+1)=f_2(g_1(x)+1);}\\\\
{f_2(0)=0;}\\ {f_2(x+1)=f_2(g_2(x+1))+1;}\\\\
{g_1(0)=0;}\\ {g_1(1)=0;}\\
{g_1(x+1+1)=g_1(x)+1;}\\\\
{g_2(0)=0;}\\ {g_2(1)=0;}\\
{g_2(x+1+1)=g_2(x)+1;}\\\hline
\end{array}
\end{array}
$$
}
\caption{Residual programs generated using Turchin's relation and using the homeomorphic embedding (Example~\ref{Example:log2proggen})}
\label{fig:progTurMSG}
\end{figure}
\end{Example}
 
\subsection{Formalization of Turchin's Relation for Call-by-Value Semantics}

In this subsection we briefly recall notions from the paper~\cite{ANNTur} where a formal proof of the well-binariness of Turchin's relation is given for the call-by-value semantics. Details are omitted. The full description of the formalization can be found in the original paper~\cite{ANNTur}. We use some notions from the paper as a basis to construct the formalization for the call-by-name semantics. 

Call stack behaviour of the programs based on the call-by-value semantics can be modelled by prefix grammars. Namely, the sequence of the call-stack configurations on every computation path can be represented as a trace generated by a prefix grammar. 

\begin{Definition}
A tuple $\langle \Upsilon, \mathbf{R}, \Gamma_0 \rangle$, where $\Upsilon$ is a finite alphabet, $\Gamma_0\in \Upsilon^+$ is an initial word, and $\mathbf{R}\subset \Upsilon^+\times \Upsilon^*$ is a finite set of rewriting rules, is called \emph{a prefix grammar} if \(R:R_l\rightarrow R_r\in \mathbf{R}\) can be applied only to words of the~form \(R_l \Phi\) (where
\(R_l\) is a prefix of the word \(R_l \Phi\) and $\Phi\in\Upsilon^*$ is arbitrary) and generates only words of the~form \(R_r \Phi\). 

Prefix grammar $\langle \Upsilon, \mathbf{R}, \Gamma_0 \rangle$ is called \emph{an alphabetic prefix grammar} if the the length of left-hand sides  of all the rules in $\mathbf{R}$ is \(1\) (only the first letter of a word is changed by any rule).

\emph{A trace} of a prefix grammar $\mathbf{G}=\langle \Upsilon, \mathbf{R}, \Gamma_0 \rangle$ is a word sequence $\{\Phi_i\}$ (finite or infinite) where $\Phi_1=\Gamma_0$ and for all $i$ $\exists R(R: R_l\rightarrow R_r \logand R\in \mathbf{R} \logand \Phi_i=R_l \Theta\logand \Phi_{i+1}=R_r \Theta)$ (where $\Theta$ is a suffix). In other words, the~elements of the trace are derived from their predecessors by applying the rewriting rules from $\mathbf{G}$.

\end{Definition}

Henceforth, the empty word is denoted by $\Lambda$.

\begin{Example}
\label{Example:CBVgram}

The following prefix grammar models the call stack behaviour of the program of Example~\ref{Example:log2prog} in the call-by-value semantics\footnote{The details of the construction may depend on the interpretation strategy.}.

\footnotesize {$$\begin{array}{|l|l|}\hline \textrm{The rules of the program} &\textrm{The rewriting rules}\\\hline
{f(0)=0;}& f\rightarrow \Lambda \\ {f(x+1)=f(g(x+1))+1;}& f\rightarrow gf\\&\\
{g(0)=0;}& g\rightarrow \Lambda\\ {g(x+1)=h(x);} & g\rightarrow h\\&\\
{h(0)=0;}& h\rightarrow \Lambda\\ {h(x+1)=g(x)+1;} & h\rightarrow g\\\hline
\end{array}$$} 

\normalsize If the initial configuration is $f(h(x))$ then the initial word $\Gamma_0$ of the grammar is $hf$. 
\end{Example}

Finally, we recall how Turchin's relation is defined for the traces generated by a prefix grammar. Given word $\Phi$ in a trace $\{\Psi_k\}$, an occurrence of letter $a$ in $\Phi$ is said to be changed with respect to a segment $[i,j]$ (where $i<j$) of the trace if some rewriting rule was applied to this occurrence of letter $a$ in the segment starting at $\Psi_i$ and ending at $\Psi_j$.

\begin{Example}
Consider the grammar given in Example~\ref{Example:CBVgram}. Let us first apply rule $h\rightarrow g$ to $\Gamma_0$, and then apply rule $g\rightarrow h$ to the result. In the following trace segment

$$\Gamma_0:hf \xrightarrow{h\rightarrow g}{}\Gamma_1:gf\xrightarrow{g\rightarrow h}{} \Gamma_2:hf$$

\noindent the occurrence of $h$ in $\Gamma_2$ is not exactly the same as the occurrence of $h$ in $\Gamma_0$ (it is rewritten by $h\rightarrow g$), so $h$ is changed in $\Gamma_2$ with respect to $[0, 2]$, while the letter $f$ in $\Gamma_2$ is unchanged with respect to $[0, 2]$.
\end{Example}

Now we formalize the definition of Turchin's relation for the function call stacks given in \cite{NemytykhBook}.

\begin{Definition}
Given a prefix grammar \(\mathbf{G}=\langle \Upsilon, \mathbf{R}, \Gamma_0 \rangle\) and a trace $\{\Gamma_k\}$ generated by $\mathbf{G}$, we say that two words $\Gamma_i$, $\Gamma_j$ in $\{\Gamma_k\}$ form \emph{a Turchin pair}  (denoted as $\Gamma_i \preceq \Gamma_j$) if $\Gamma_i = \Phi\Theta_0$, $\Gamma_j = \Phi\Psi\Theta_0$ and the suffix $\Theta_0$ is not changed in the trace segment $[i,j]$.
\end{Definition}

In order to develop this formalization for call-by-name languages, below we introduce an extension of the class of the prefix grammars.

\section{Multi-Layer Prefix Grammars}
\label{section:MLPG}

\subsection{Motivation}

Turchin's relation considers every stack as a word consisting of function names and ignores the arguments of the function calls. Hence, from the point of view of the relation, call-stack configurations form a trace generated by some grammar. So the problem arises: given a program, what class of grammars can generate words that correspond to the call-stack configurations generated by the program? 

For the call-by-value semantics, the construction of the grammar is very straightforward (as can be seen in Example~\ref{Example:CBVgram}). But in the call-by-name semantics, besides the active part of the stack, function calls may appear in the passive part of the configuration. Such function calls make it impossible to predict how the call stack will be transformed if we observe only its active part.

\begin{Example}
\label{Example:CBNpath}

Given a path from the tree of Example~\ref{Example:log2prog}, we see that applications of the same rule $g(x+1)=h(x)$ resulted in different stack transformations. In the first case, the call $g(x_1)$ was popped from the passive part of the configuration.

\begin{center}
\begin{tabular}{l|l|l}
\multicolumn{1}{c|}{Computation path}&\multicolumn{1}{c|}{Applied rule}&\multicolumn{1}{c}{Call stack}\\
\hline 
\begin{tabular}{l}
{\footnotesize${f(h(x))}$}\\
{\footnotesize${f(g(x_1)+1)}$}\\
\footnotesize${f(g(g(x_1)+1))+1}$\\
{\footnotesize${f(h(g(x_1)))+1}$}\\
{\footnotesize${f(h(h(x_2)))+1}$}\\
\end{tabular}
& 
\begin{tabular}{l}\\
{\footnotesize${h(x+1)=g(x)+1}$}\\
{\footnotesize${f(x+1)=f(g(x+1))+1}$}\\
{\footnotesize${g(x+1)=h(x)}$}\\
{\footnotesize${g(x+1)=h(x)}$}
\end{tabular}
& 
\begin{tabular}{l}
\footnotesize{${hf}$}\\
\footnotesize${f}$\\
\footnotesize${gf}$\\
\footnotesize${ghf}$\\
\footnotesize${hhf}$\\
\end{tabular}

\end{tabular}
\end{center}

\end{Example}

Example~\ref{Example:CBNpath} shows that the only way to build a consistent grammar model of the call stack behaviour is to take into account the passive part of the configuration. 

Every configuration is a tree of function calls, both active and passive. Given a configuration, the active call stack forms a path in the configuration tree starting at the root.

\begin{Example}
\label{Example:treeAck}

Given the configuration $${b(x_1+1,b(d(x_1),d(x_2))+1)}$$ we apply the following rule $${b(x_1+1,x_2+1) = b(d(b(x_1,x_2+1)),x_2)}$$ to it. The tree representations of the results are given in Figure~\ref{fig:treeAck}.

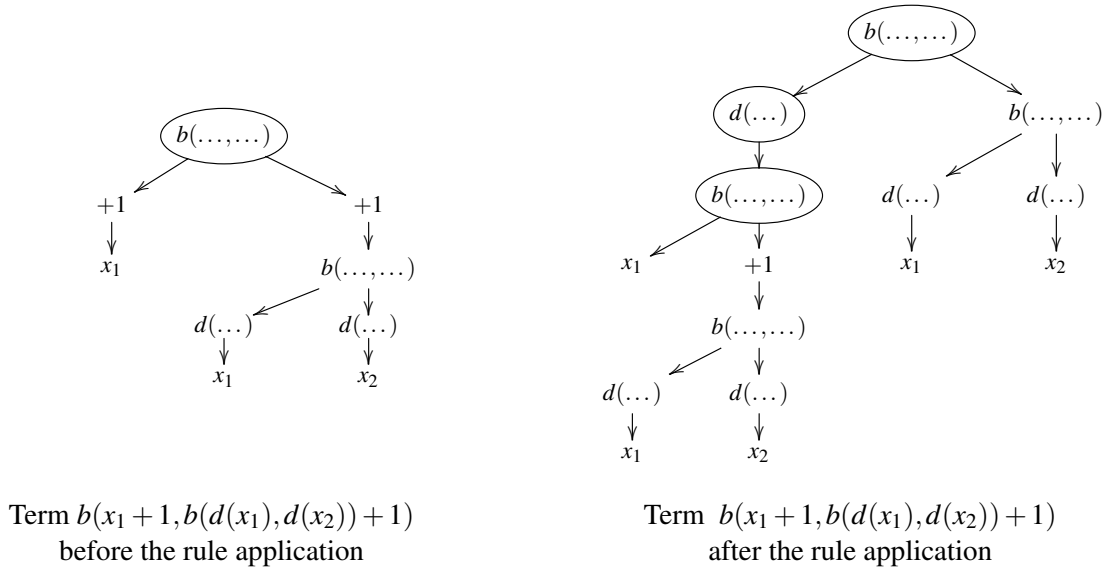
\begin{figure}[h]

\footnotesize
\begin{tabular}{ccc}
$$
\xymatrix @-5mm {\\\\
&&&*++[o][F]{{b(\dots,\dots)}}\ar[dl]\ar[dr]\\
&&{+1}\ar[d]&&{+1}\ar[d]\\
&&*+{x_1}&& *+{{b(\dots,\dots)}}\ar[dl]\ar[d]\\
&&&*{{d(\dots)}}\ar[d]&*{{d(\dots)}}\ar[d]\\
&&&*+{x_1}&*+{x_2}
}
$$ 
&\qquad \qquad \qquad &
$$
\xymatrix @-5mm {
&&*++[o][F]{{{b(\dots,\dots)}}}\ar[dl]\ar[dr]
\\
&*++[o][F]{d(\dots)}\ar[d]&& *+{{b(\dots,\dots)}}\ar[dl]\ar[d]
\\
&*++[o][F]{{b(\dots,\dots)}}\ar[dl]\ar[d]&*+{{d(\dots)}}\ar[d]& *+{{d(\dots)}}\ar[d]
\\
*+{x_1}&{+1}\ar[d]&*+{x_1}&*+{x_2}\\
& *+{{b(\dots,\dots)}}\ar[dl]\ar[d]\\
*+{{d(\dots)}}\ar[d]&*+{{d(\dots)}}\ar[d]\\
*+{x_1}&*+{x_2}
}
$$\\
\\\normalsize
\begin{tabular}{l}{Term }${b(x_1+1,b(d(x_1),d(x_2))+1)}$ \\\multicolumn{1}{c}{{before the rule application}}
\end{tabular}&&\normalsize \begin{tabular}{l}{Term } ${b(x_1+1,b(d(x_1),d(x_2))+1)}$ \\\multicolumn{1}{c}{{after the rule application}}
\end{tabular}
\end{tabular}

\caption{Tree forms of the terms of Example~\ref{Example:treeAck}. The calls in the active call stack are in ovals.}
\label{fig:treeAck}
\end{figure}
\end{Example}

Based on the observations above, we use the following assumptions to construct the grammar models for programs based on the call-by-name semantics.

\begin{enumerate}
\item A configuration can be considered as a tree of calls, and the active call stack --- as a path in the tree. We use a set of labels $\mathbf{S}$ with partial order $\triangleleft$ for denoting the positions of the function calls in the tree.

\item Every call in the stack is modelled by a pair $<\textrm{NAME, }\textrm{LABEL}>$, where $\textrm{LABEL}\in\mathbf{S}$.

\item Every configuration is represented as a word $\Gamma \$ \Delta$ consisting of the two parts separated by the symbol $\$$. The structure of the active stack is placed in $\Gamma$ and is linearly ordered w.r.t. labels, the function calls in the passive part of the configuration are placed in $\Delta$.
\end{enumerate}

\subsection{Formalization}
\label{subsection:forMLPG}

Let $\Upsilon$ be a finite alphabet. Let $\mathbf{S}$ be \emph{a label set} and $\triangleleft$ be a strict (non-reflexive) partial order relation over $\mathbf{S}$. We denote the labels from $\mathbf{S}$ by the letters $s$, $t$ (maybe with subscripts). Let us say that $s_1$ is a child of $s_0$ w.r.t. $\mathbf{S'}\subseteq \mathbf{S}$ (denoted by $s_1=\child(s_0)[\mathbf{S'}]$) if $s_0\triangleleft s_1$, $s_0\in\mathbf{S'}$, $s_1\in\mathbf{S'}$ and there is no such $s_2\in\mathbf{S'}$ that $s_0\triangleleft s_2$ and $s_2 \triangleleft s_1$. The inverse for the child relation is the parent relation. Given a set $\mathbf{S'}\subseteq \mathbf{S}$ and a label $t\in\mathbf{S}\setminus\mathbf{S'}$, we call $t$ \emph{a fresh label} w.r.t. $\mathbf{S'}$ if $\mathbf{S'}$ contains neither descendants nor ancestors of label $t$\footnote{In most cases, we assume that $\mathbf{S'}$ is a set of all previously used labels, hence there is no need to write it in the square brackets in expressions like $\child(s_0)[\mathbf{S'}]$.}. 

Informally, the labels can be considered as nodes of trees with unbounded branching, then the child--parent relation has its usual meaning. 

Henceforth, the set of finite sequences of pairs $\{\langle a, s_i\rangle| a\in \Upsilon \logand s_i\in\mathbf{S}\}^*$ is denoted by $\lw$. Elements of $\lw$ are called \emph{layered words}, and are denoted by Greek capitals $\Gamma$, $\Delta$, $\Phi$, $\Psi$, $\Xi$, $\Theta$. If $\langle a_1, s_1 \rangle\dots\langle a_n, s_n \rangle$ is a layered word, the corresponding \emph{plain word} is defined as $a_1\dots a_n$. 

If $\Phi$ is a layered word, $|\Phi|$ stands for the number of the pairs in $\Phi$ and $\Phi[i]$ stands for the $i$-th pair. For the sake of brevity, layered word $\langle a_1, s_0\rangle \dots \langle a_n, s_0\rangle$ can be also written as $\langle a_1 \dots a_n, s_0\rangle$ (thus, $a \langle s_0 \rangle$ is an equivalent form for $\langle a, s_0\rangle$). 

Expression $\Phi\langle s_0\rangle$ denotes the maximal subsequence of $\Phi$ containing only the pairs labelled with $s_0$. Expression $\Phi\langle \ovl{$s_0$}\rangle$ denotes the maximal subsequence of $\Phi$ not containing the pairs labelled with $s_0$. The set of all labels in $\Phi$ is denoted by $\mathbf{S}_{\Phi}$.

\begin{Example}
Let $\Phi=\langle a_1, s_1\rangle \langle a_2, s_1\rangle \langle a_3, s_2\rangle \langle a_4, s_4\rangle \langle a_5, s_1\rangle \langle a_6, s_3\rangle\langle a_7, s_4\rangle$. Then $\Phi\langle s_1\rangle = \langle a_1 a_2 a_5, s_1\rangle$, $\Phi\langle \ovl{$s_1$}\rangle = \langle a_3, s_2\rangle \langle a_4, s_4\rangle \langle a_6, s_3\rangle\langle a_7, s_4\rangle$. 

If $s_1\triangleleft s_2\triangleleft s_3$, $s_1\triangleleft s_4$ (and $s_4$ is fresh w.r.t. $\{s_2, s_3\}$, then the layered word $\Phi$ can be represented as the following tree:

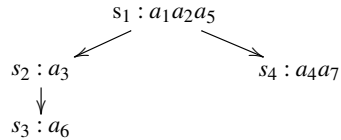
\UseAllTwocells
\footnotesize
\[
\xymatrix @-5mm {
&*+{\mathrm{s_1}: a_1 a_2 a_5}\ar[dr]\ar[dl]
\\
*+{s_2: a_3}\ar[d]&& *+{s_4: a_4 a_7}
\\
*+{s_3: a_6}
}
\] 

\normalsize The order of the letters in $\Phi$ does matter for the tree representation only if the letters have the same label\footnote{Hence, both the word $\Phi$ and, for example, word $\langle a_6, s_3\rangle \langle a_3, s_2\rangle \langle a_1, s_1\rangle \langle a_4, s_4\rangle \langle a_7, s_4\rangle \langle a_2, s_1\rangle \langle a_5, s_1\rangle$ are presented by the same tree above.}.
 
\end{Example} 

Given a label $s_i$ and natural numbers $K_1$ and $K_2$, we define a set of \emph{layer functions} w.r.t. label $s_i$, $\mathfrak{F}^{s_i}_{K_1, K_2}: \lw\rightarrow \lw$, as a minimal set of functions containing all compositions of $K_1$ elementary functions, which are:

\begin{enumerate}
\item Append $\App^{s_j}[\Psi]$ (where $s_j\in \mathbf{S}$, $\Psi\in\Upsilon^*$): given a layered word $\Phi$, $\App^{s_j}[\Psi](\Phi)$ is the word $\Phi \Psi\langle s_j\rangle$ such that $s_j$ is a child of $s_i$ w.r.t. $\mathbf{S}_{\Phi}\cup\{s_j\}$, $s_j$ is fresh w.r.t. $\mathbf{S}_{\Phi}\setminus \{s_i\}$, and $|\Psi|\leq K_2$. 

For example, if $\App^{s_1}[g]\in \mathfrak{F}^{s_0}_{1, 1}$ and $s_0\triangleleft s_1$, then $${\App}^{s_1}[g](\langle f, s_0\rangle \langle g, s_1\rangle)=\langle f, s_0\rangle \langle g, s_1\rangle \langle g, s_1\rangle$$

 The appending operation, if considered as a tree transformation, appends some new letters to an existing node (Figure~\ref{fig:appendtree}).

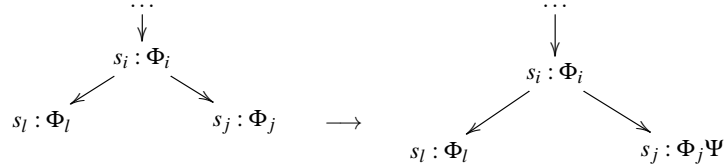
\begin{figure}[ht]
\footnotesize
\begin{center}
\begin{tabular}{lll}
$$\entrymodifiers={+}
\xymatrix @-5mm { 
&{\dots}\ar[d]\\
&{s_i: \Phi_i}\ar[dl]\ar[dr]
\\
{s_l:\Phi_l}&& {s_j: \Phi_j}
}
$$ 
&$$\xymatrix @-3mm { 
\\\\
{\longrightarrow}}
$$ &
$$\entrymodifiers={+}
\xymatrix @-3mm { & {\dots}\ar[d]\\
&{s_i: \Phi_i}\ar[dl]\ar[dr]
\\
{s_l:\Phi_l}&& {s_j: \Phi_j \Psi}
}$$
\end{tabular}
\end{center}
\caption{Function $\App^{s_j}[\Psi]$ as a tree operation}
\label{fig:appendtree}
\end{figure}

\item Insert $\Lift^{s_j}[\Psi\langle s_k\rangle]$ (where $s_j, s_k\in \mathbf{S}$, $\Psi\in\Upsilon^*$): given $\Phi$ with a non-empty $\Phi\langle s_j\rangle$, where $s_j$ is a child of $s_i$ w.r.t. $\mathbf{S}_{\Phi}$, $\Lift^{s_j}[\Psi\langle s_k\rangle](\Phi)$ is the word $\Phi\Psi\langle s_k\rangle$ where $|\Psi|\leq K_2$ and $s_k$ is a child of $s_i$ w.r.t. $\mathbf{S}_{\Phi}\cup{\{s_k\}}$, $s_k$ is fresh w.r.t. $\mathbf{S}_{\Phi}\setminus \{s_i\}$ and $s_j$ is a child of $s_k$ w.r.t. $\mathbf{S}_{\Phi}\cup{\{s_k\}}$. 

For example, if $\Lift^{s_1}[gf\langle s_{2}\rangle]\in \mathfrak{F}^{s_0}_{1, 1}$\footnote{This condition implies that $s_2\triangleleft s_1$ and $s_0\triangleleft s_2$.} and $s_0\triangleleft s_1$, then $${\Lift}^{s_1}[\langle gf, s_{2}\rangle](\langle f, s_0\rangle \langle g, s_1\rangle)=\langle f, s_0\rangle \langle g, s_1\rangle \langle gf, s_{2}\rangle$$

The insert operation differs from the append operation only by introduction of an unused child label $s_k$, which marks the newly appended word $\Psi$. In tree terms, this operation inserts a new node between the nodes labelled by $s_i$ and $s_j$ (Figure~\ref{fig:inserttree}).

\begin{figure}[ht]
\footnotesize
\begin{center}
\begin{tabular}{lll}
$$\entrymodifiers={+}
\xymatrix @-5mm { 
&{\dots}\ar[d]\\
&{s_i: \Phi_i}\ar[dl]\ar[dr]
\\
{s_l:\Phi_l}&& {s_j: \Phi_j}
}
$$ 
&$$\xymatrix @-3mm { 
\\\\
{\longrightarrow}}
$$ 
&
$$\entrymodifiers={+}
\xymatrix @-3mm { & {\dots}\ar[d]\\
&{s_i: \Phi_i}\ar[dl]\ar[dr]
\\
{s_l:\Phi_l}&& {{s_k:\Psi}}\ar[d]
\\
&& {s_j:\Phi_j}
}$$
\end{tabular}
\end{center}
\caption{Function $\Lift^{s_j}[\Psi\langle s_k\rangle]$ as a tree operation}
\label{fig:inserttree}
\end{figure}
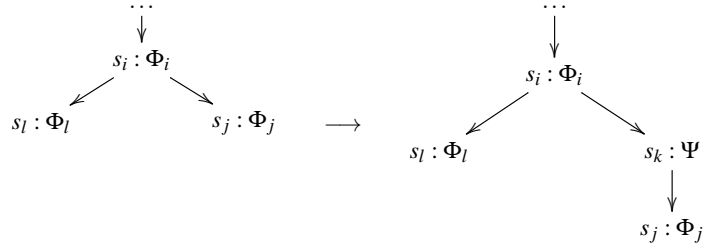

\item Deleting $\Del^{s_j}$ (where $s_j\in \mathbf{S}$): given $\Phi$ with a non-empty $\Phi\langle s_j\rangle$, $s_j=\child(s_i)$ w.r.t. $\mathbf{S}_{\Phi}$, $\Del^{s_j}$ erases $\Phi\langle s_j\rangle$ from $\Phi$ together with all $\Phi\langle t\rangle$ for which $s_j\triangleleft t$.

For example, if ${\Del}^{s_{01}}\in \mathfrak{F}^{s_0}_{1, 1}$ and $s_0\triangleleft s_{01}$, $s_{02}$ is incomparable with $s_{01}$, then $${\Del}^{s_{01}} (\langle d, s_{01}\rangle \langle d, s_{02}\rangle)=\langle d, s_{02}\rangle$$ 

In tree terms, this operation deletes the subtree, whose uppermost node is labelled by $s_j$ (Figure~\ref{fig:deletetree}).

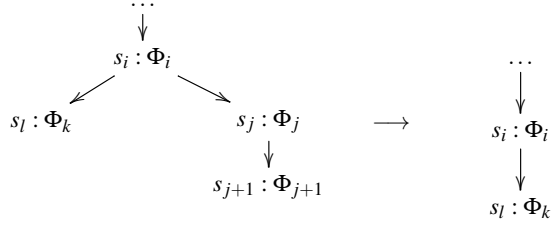
\begin{figure}[ht]
\footnotesize
\begin{center}
\begin{tabular}{llll}
$$\entrymodifiers={+}
\xymatrix @-5mm { 
&{\dots}\ar[d]\\
&{s_i: \Phi_i}\ar[dl]\ar[dr]
\\
{s_l:\Phi_k}&& {{s_j: \Phi_j}}\ar[d]\\
&& {{s_{j+1}: \Phi_{j+1}}}}
$$ 
&
$$\xymatrix @-3mm { 
\\\\
{\longrightarrow}}
$$&&
$$\entrymodifiers={+}
\xymatrix @-3mm { \\ {\dots}\ar[d]\\
{s_i: \Phi_i}\ar[d]
\\
{s_l:\Phi_k}
}$$
\end{tabular}
\end{center}
\caption{Function $\Del^{s_j}$ as a tree operation}
\label{fig:deletetree}
\end{figure}

\item Copying $\Copy^{s_j}$ (where $s_j\in \mathbf{S}$): given $\Phi$ with a non-empty $\Phi\langle s_j\rangle$, $s_j=\child(s_i)$ w.r.t. $\mathbf{S}_{\Phi}$, $\Copy^{s_j}$ appends $\Phi\langle s_k\rangle$ to $\Phi$, where $s_k$ is a child of $s_i$ w.r.t. $\mathbf{S}_{\Phi}\cup\{s_j\}$, $s_j$ is fresh w.r.t. $\mathbf{S}_{\Phi}\setminus \{s_i\}$, and then it appends all subsequences $\Phi\langle s_l\rangle$ labelled by the children of $s_j$ and labels them by fresh children of $s_l$ and so on until all the sequences $\Phi\langle t\rangle$, where $s_j\triangleleft t$, are copied exactly once. 

For example, if ${\Copy}^{s_{01}}\in \mathfrak{F}^{s_0}_{1, 1}$ and $s_0\triangleleft s_{01}$, then $${\Copy}^{s_{01}} (\langle d, s_{01}\rangle)=\langle b, s_0\rangle \langle d, s_{01}\rangle\langle d, s_{02}\rangle,$$ where $s_{02}$ is incomparable with $s_{01}$.

In tree terms, this operation creates a copy of the subtree, whose uppermost node is labelled by $s_j$ (Figure~\ref{fig:copytree}).

\begin{figure}[ht]
\footnotesize
\begin{center}
\begin{tabular}{lll}
$$\entrymodifiers={+}
\xymatrix @-5mm { 
&{\dots}\ar[d]\\
&{s_i: \Phi_i}\ar[dl]\ar[dr]
\\
{s_l:\Phi_l}&& {s_j: \Phi_j}\ar[d]\\
&&{s_{j+1}: \Phi_{j+1}}
}
$$ 
&$$\xymatrix @-3mm { 
\\\\
{\longrightarrow}}
$$
&
$$\entrymodifiers={+}
\xymatrix @-3mm { & {\dots}\ar[d]\\
&{s_i: \Phi_i}\ar[dl]\ar[d]\ar[dr]
\\
{s_l:\Phi_l}& {s_j:\Phi_j}\ar[d] & {{s_k: \Phi_j}}\ar[d]\\
&{s_{j+1}: \Phi_{j+1}}&{{s_{k+1}: \Phi_{j+1}}}
}$$
\end{tabular}
\end{center}
\caption{Function $\Copy^{s_j}$ as a tree operation}
\label{fig:copytree}
\end{figure}
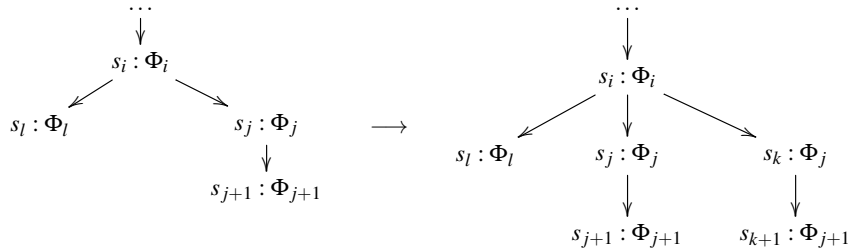

\end{enumerate}

Section~\ref{section:MCSBMLG} below shows how the elementary layer functions model the call stack transformations.

\begin{Definition}
Let us consider a tuple $\mathbf{G}=\langle \Upsilon, \mathbf{S},\mathbf{R}, \mathfrak{F}^v_{K_1, K_2},\Gamma_0\$\Delta_0\rangle$ where $\Gamma_0$ and $\Delta_0$ are layered words over $\Upsilon\times \mathbf{S}$ such that for every $\Gamma_0[i]=\langle a_i, s_i\rangle$ and $\Gamma_0[j]=\langle a_j, s_j\rangle$, if $j>i$ then $s_j\triangleleft s_i$ or $s_j=s_i$, $\$$ is a special symbol, $\$\notin\Upsilon$, and $\mathfrak{F}^v_{K_1, K_2}$ is a finite set of layer function forms where $v$ runs over the label set $\mathbf{S}$. For every $\mathbf{G}$-word $\Gamma\$\Delta$, where $\Gamma$ and $\Delta$ are words in $\lw$, we call $\Gamma$ \emph{the visible layer}, and we call $\Delta$ \emph{the invisible layer} of $\Gamma\$\Delta$.

Let all rewriting rules from $\mathbf{R}$ have one of the following forms:

\begin{itemize}
\item Simple rule: $$\Xi\langle a, s_i\rangle \Theta\$\Psi\rightarrow \Phi \Theta \$ F^{s_i}(\Psi),$$ where all the letters of $\Phi$ are labelled either by $s_i$ or by fresh descendants of $s_i$, $F^{s_i} \in \mathfrak{F}^{s_i}$. 

\item Pop rule: for $\Psi\langle s_j \rangle$ --- the maximal subsequence of $\Psi$ marked by some $s_j=\child(s_i)\in \mathbf{S}$, $$\Xi\langle a, s_i\rangle \Theta\$\Psi\rightarrow \Psi\langle s_j \rangle \Phi\Theta \$ F^{s_i}(\Psi),$$  where all the letters of $\Phi$ are labelled either by $s_i$ or by fresh descendants of $s_i$, $F^{s_i} \in \mathfrak{F}^{s_i}$. In a pop rule, we may specify $s_j$, but there are no ways to specify $\Psi\langle s_j \rangle$.
\end{itemize}    

Such a grammar $\mathbf{G}$ is called {a multi-layer prefix grammar}. $K_2$ is called {the maximal rewrite depth}, $K_1$ is called {the maximal replication index}. A sequence of $\mathbf{G}$-words starting at $\Gamma_0\$\Delta_0$ that are transformed by the rules from $\mathbf{R}$ is called {a trace} of $\mathbf{G}$. 

If any rule of such a grammar changes only one letter of the visible layer (thus, $\Xi=\Lambda$), then the multi-layer prefix grammar is \emph{alphabetic}. 
\end{Definition}

\begin{Definition}
Let $\Phi\Theta\$\Delta_i$ be the $i$-th $\mathbf{G}$-word in a trace $\{\Gamma_k\$\Delta_k\}$ generated by an alphabetic multi-layer prefix grammar $\mathbf{G}$. If $\Gamma_j\$\Delta_j=\Psi\Theta\$\Delta_j$ ($j>i$) then we say that $\Psi$ is \emph{a derivative prefix} (or simply \emph{a~derivative}) of $\Phi$ (denoted by $\deriv(\Phi)$).
\end{Definition}

Now we can prove some simple propositions about the multi-layer grammars.

\begin{Prop}
\label{Prop:VisOrder}

The following properties hold.
\vspace{1pt}
\begin{enumerate}
\item Given an alphabetic multi-layer prefix grammar $\mathbf{G}$ and a word $\Gamma\$\Delta$ generated by $\mathbf{G}$, for every  $\Gamma[i]=\langle a_i, s_i\rangle$, $\Gamma[j]=\langle a_j, s_j\rangle$ if $i>j$ then either $s_i=s_j$ or $s_i\triangleleft s_j$.

\item Let an alphabetic multi-layer grammar $\mathbf{G}$ generate $\Gamma\$\Psi$ with the non-empty $\Psi\langle s_1\rangle$ and $\Psi\langle s_1\rangle$ such that $s_1$ and $s_2$ are unequal and incomparable. Then derivatives of $\Psi\langle s_1\rangle$ and $\Psi\langle s_2\rangle$ cannot occur in the visible part of the same word.

\item Let $\mathbf{G}$ generate a word $\Xi=\langle a, s_i\rangle \Gamma \langle b,s_j\rangle \Theta \$ \Psi$. Given the trace containing $\Xi$, any occurrence of $\deriv(\langle a,s_i \rangle)$ in the visible layer of the word in the trace precedes all occurrences of $\deriv(\langle b,s_j\rangle)$.
\end{enumerate}
\end{Prop}

\begin{proof}
1. By the definition, the initial word $\Gamma_0\$\Delta_0$ satisfies the stated property. Given a word $\langle a, s_i\rangle \Gamma\$\Delta$ satisfying the stated property, a non-pop rule can only prepend descendants of $\langle a, s_i\rangle$ to $\Gamma$, and a pop rule can only append descendants of $\langle a, s_i\rangle$ to $\Gamma$. Therefore, all the words generated from $\langle a, s_i\rangle \Gamma\$\Delta$ must satisfy the stated property.  

\vspace{5pt}
2. If $s_1$ and $s_2$ are incomparable then the labels of the derivatives of $\Psi\langle s_1\rangle$ and $\Psi\langle s_2\rangle$ are also incomparable. And the stated property follows from the case 1.

\vspace{5pt}
3. Given a trace generated by $\mathbf{G}$, $\deriv(\langle b,s_j\rangle)$ can be generated only after word $\langle b,s_j\rangle \Theta \$ \Psi'$ occurring in the trace. Word $\langle b,s_j\rangle \Theta \$ \Psi'$ can contain $\deriv(\langle a,s_i\rangle)$ only in the invisible layer. All $\deriv(\langle b,s_j\rangle)$ are appended to the end of the invisible layer (by the definitions of $App^{s_k}[\Psi]$, $\Lift^{s_k} [\Psi\langle s_u \rangle]$, $\Copy^{s_k}$). They are marked by the labels, which are not less (but may be incomparable) than the labels of $\deriv(\langle a,s_i \rangle)$. Hence, after applying the pop rules, $\deriv(\langle b,s_j\rangle)$ follows $\deriv(\langle a,s_i\rangle)$ in the visible layer.
\end{proof}

To specify the notion of a language generated by a multi-layer prefix grammar, we add one or several halting rules $R^{[i]}_{\nil}$ to the set $\mathbf{R}$. They are of the usual form, but halt the computation path. 

\vspace{10pt}
\begin{Definition}
\emph{A language generated by} $\mathbf{G}$ is the set of all the words $A\in\Upsilon^*$ where $A$ is a plain word corresponding to the visible layer $\Gamma$ of some layered word $\Gamma\$\Delta$, such that:
\begin{itemize}
\item $\Gamma\$\Delta$ is generated in a finite trace of multi-layer prefix grammar $\mathbf{G}$;
\item $\Delta$ is the result of an application of a halting rule.
\end{itemize}
\end{Definition}

The rules $R^{[i]}_{\nil}$ can be considered as function definitions that cause the printing side effect. When some $R^{[i]}_{\nil}$ is applied, we output the stack configuration where it happened.

\begin{Theorem}
\label{Theorem:CompComp}
Every recursively enumerable set can be generated by a multi-layer prefix grammar. 
\end{Theorem}

The proof of Theorem~\ref{Theorem:CompComp} can be found in Appendix (proof~\ref{proof:thcompcomp}). Every rule in the grammar constructed in the proof changes the two first letters of the visible part of a $\mathbf{G}$-word. In terms of computations, this means that the two calls belonging to the active stack, which is modelled by the visible part, are evaluated by a single action. Usually, only a single call from the top of the stack is evaluated, so the modelling grammar is alphabetic. In the case of the plain prefix grammars, the class of the generating alphabetic prefix grammars defines the regular languages as well as the class of all prefix grammars \cite{Caucal}. In the case of the alphabetic multi-layer prefix grammars, the situation changes drastically. We can informally compare their power to 1-state Turing machines, although the grammars modelling Turing machines use only the insert layer functions, while the copy and append functions are left aside. 

\begin{Prop}
\label{Prop:CompAlphabeticMLPG}

Alphabetic multi-layer prefix grammars are strictly stronger\footnote{They generate all the languages that can be generated by these models, and some languages that cannot be generated by these two models.} than both context-free grammars and one-state Turing machines.
\end{Prop}

The proof of Proposition~\ref{Prop:CompAlphabeticMLPG} can be found in Appendix (proof~\ref{proof:compalph}).

Given a word $w=w[1]w[2]\dots w[N]$, \emph{the inverse word of} $w$ is the word $w[N]\dots w[2] w[1]$ (denoted by $\inv(w)$).

\begin{Prop}
There is no rule-set $\mathbf{R}$ such that any alphabetic multi-layer prefix grammar with the rule-set $\mathbf{R}$ and the initial word from $w\in\{a,b\}^*$ can construct $\inv(w)$.

\end{Prop}
\begin{proof}
Let $w$ be modelled by the initial word $\Phi_0 \langle a_i, s_i\rangle \Phi_1 \langle b_j, s_j\rangle \$ \Lambda$. The word $\inv(w)$ is modelled by $\langle b_j, t_j\rangle \inv(\Phi_1) \langle a_i, t_i\rangle \inv(\Phi_0)\$\Lambda$. According to Proposition~\ref{Prop:VisOrder}, because $\langle b_j, t_j\rangle=\deriv(\langle b_j, s_j\rangle)$, $\langle a_i, t_i\rangle=\deriv(\langle a_i, s_i\rangle)$, such a word cannot appear in any trace generated by any alphabetic multi-layer prefix grammar. 
\end{proof}

\normalsize This proof shows that the class of the alphabetic multi-layer prefix grammars does not coincide with the classes of the tree automata grammars and linear indexed grammars \cite{HopcroftUlmanBook}. Informally, this class contains grammars that are able to generate very long words of a rather simple structure.

\section{Modelling Call Stack Behaviour by Multi-Layer Grammars}
\label{section:MCSBMLG}

We borrow the notions of $f$-function and $g$-function from \cite{SorensenMsTh} and use them in the following sense. An $f$-function is a function whose definition consists of one rule with trivial patterns (e.\,g., if ${h_1}$ is defined as ${h_1(x_1,x_2)=b(x_2,h_2(x_1+1))}$ then ${h_1}$ is an $f$-function). A $g$-function is a function with non-trivial patterns in the definition (e.\,g., ${h_2(x+1)=b(h_2(x),h_2(x))}$ is a definition of the $g$-function).

In order to get a grammar from a program, we treat every configuration generated by the unfolding as a tree, whose nodes are named by function or constructor names and leaves contain no function calls. First, we mark every function name in the tree by a superscript depending on the state of the function call. If the call is ready to be evaluated without evaluation of other calls, the function name is marked as ``ready'' (by $+$ in the superscript). Otherwise, the function name is marked as ``unready''. Hence, the call names of all $f$-functions are always marked as ``ready'', while the call names of $g$-functions are marked as ``ready'' if the patterns of the functions can be matched without evaluating another call\footnote{For programs in language $\mathbb{L}$, we always can determine all the calls that are ready to be unfolded due to simplicity of the patterns. In languages with complex pattern matching (e.\,g., Refal~\cite{TurRefal5}), that can be done only if one knows the strategy of the pattern matching applied in the interpreter.}.

\vspace{+10pt}
Then we delete all the nodes containing static data\footnote{In some cases, this action can transform the tree into a forest. For example, that can happen if the configuration is ${cons(h_1(x),cons(h_2(x),Nil))}$. To avoid these cases, we always assume that the transformed tree has a root, but the root is a ``virtual'' function call, which is always present in the $\mathbf{G}$-word corresponding to the tree and is denoted by $\$$.}. The remaining nodes are given the layer labels. If some node $T$ is a descendant of a node $W$, the label of $T$ is greater than the label of $W$. Otherwise the labels are incomparable. 

\vspace{+10pt}
Finally, we find all the nodes containing the unready call names with a single child. The child of such a node is given the label of the node. And then, all the nodes with the same labels are merged: data from the ancestor nodes are placed in the merged node after the data from their descendants. 

\vspace{+10pt}
The resulting tree is a tree form of the corresponding layered word.

\vspace{+10pt}
\begin{Example}

Given the term ${b(d(x_1)+1,b(x_1,d(d(x_2)))+1)}$, we transform it to a layered word. 
All the steps of the transformation are given in Figure~\ref{fig:treetolayeredword}.

\vspace{+10pt}
First, we mark the calls as ``ready''(with $+$ in the superscript) and ``unready''(with $-$ in the superscript), and delete the nodes with the static data. The only function call in the configuration which is ready to be evaluated without unfolding is the outermost call of ${b}$. The calls ${d(x_1)}$ and ${d(x_2)}$ require unfolding (which generates restrictions on $x_1$ and $x_2$), but they do not require evaluation of other calls, so they are also marked as ready. All the other calls are marked as unready.

\vspace{+10pt}
Then we assign the layer labels in the resulting tree of the marked call names. The tree below shows that $s_0\triangleleft s_1$, $s_0\triangleleft s_2$, $s_2\triangleleft s_3\triangleleft s_4$.

\vspace{+10pt}
After that, we find all the nodes with a single child whose call name is marked as unready. The child of such a node is given the layer label of its parent. Then we merge all the nodes having the same label and prepend the data in the descendant nodes to the data in the ancestor nodes. In the second tree in Figure~\ref{fig:treetolayeredword}, the two nodes will be given their ancestor's layer label: the nodes labelled by $s_{3}$ and by $s_{4}$. 

\vspace{+10pt}
Finally, the node containing the name of the call in the active stack is extracted. In Figure~\ref{fig:treetolayeredword}, the node is $s_0$ containing function name $b$. This name together with the node label take a place in the visible part of the layered word; data from all the other nodes in the tree (namely, nodes $s_1$ and $s_2$) take a place in the invisible part. 

\begin{figure}[ht]
\footnotesize
\begin{minipage}{0.76\textwidth}
\begin{tabular}{cccc}
$$
\xymatrix @-5mm {
&*+{b(\dots,\dots)}\ar[dl]\ar[d]
\\
*+{+1}\ar[d]&*+{+1}\ar[d]
\\
*+{d(\dots)}\ar[d]& *+{b(\dots,\dots)}\ar[d]\ar[dr]
\\
*+{x_1}&*+{x_1}&*+{d(\dots)}\ar[d]
\\
&&*+{d(\dots)}\ar[d]\\
&&*+{x_2}}
$$
&
$$
\xymatrix @-5mm {
&*+{s_{0}: b^+}\ar[dl]\ar[dd]
\\
*+{s_{1}: d^+}
\\&*+{s_{2}: b^-}\ar[d]\\
&*+{s_{3}: d^-}\ar[d]\ar@{.}@/_2pc/[u]\\
&*+{s_{4}: d^+}\ar@{.}@/_2pc/[u]
}
$$ &
$$
\xymatrix @-5mm {
\\\\
&&*+[o][F]{s_{0}: b}\ar[dl]\ar[dd]
\\
&*+{s_{1}: d}\\
&&*+{s_{2}: ddb}\\
}
$$
\\
\\
\begin{tabular}{c}\quad\normalsize\textbf{(a)}\end{tabular} &
\begin{tabular}{c}\quad\quad\normalsize\textbf{(b)}\end{tabular}&
\begin{tabular}{c}\quad\quad\quad\quad\normalsize\textbf{(c)}\end{tabular}
\end{tabular}
\end{minipage} 
\begin{minipage}{0.22\textwidth}\footnotesize
$$\begin{array}{|l|}
\hline
{\mathrm{Program}} \\\hline\\
{b(0,x_2)=1;}\\ 
{b(x_1,0)=x_1;}\\ 
{b(x_1+1,x_2+1)=}\\
\qquad{b(d(b(x_1,x_2+1)),}\\
\qquad\quad{x_2);}
 \\\\
{d(0)=0;}\\
{d(x+1)=d(x)+1+1;}\\\hline
\multicolumn{1}{c}{}\\
\multicolumn{1}{c}{}\\
\multicolumn{1}{c}{}\\
\multicolumn{1}{c}{}\\
\multicolumn{1}{c}{}\\\end{array}$$
\end{minipage}

\caption{Steps transforming the term ${b(d(x_1)+1,b(x_1,d(d(x_2)))+1)}$ from a tree to layered word $\langle b, s_0\rangle \$ \langle d, s_{01}\rangle \langle ddb, s_{02}\rangle$: 
\textbf{(a)} the initial term in the tree form; 
\textbf{(b)} the call names are marked, the static data are deleted, the layer labels are assigned;
\textbf{(c)} the layer labels are merged, the active part is extracted.}

\label{fig:treetolayeredword}
\end{figure}
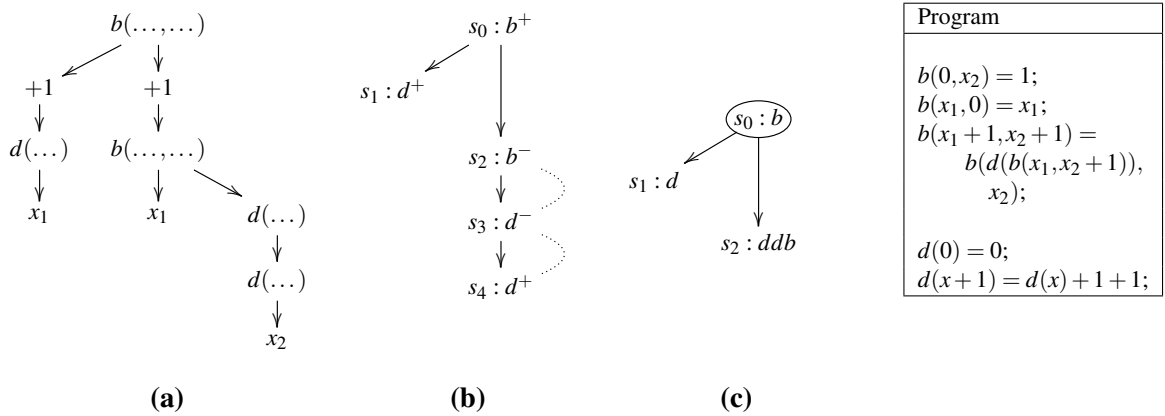

\end{Example}

The layer functions used for defining the multi-layer grammars (see Section~\ref{section:MLPG}) describe one-step actions transforming the passive parts of the function call stacks in language $\mathbb{L}$ (and the other call-by-name programming languages).
\begin{enumerate}
\item Appending ${\App}^{s_j}[\Psi]$ models adding an unready function call $\Psi$ (or a sequence of such function calls) to the passive part of the given configuration. None of the function calls from $\Psi$ can be evaluated immediately. E.\,g., as shown in Figure~\ref{fig:AppProg}.

\begin{figure}[ht]
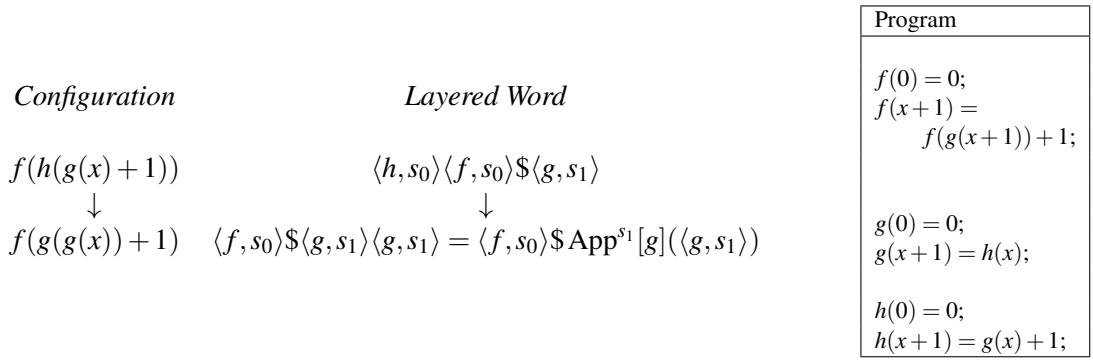

\begin{minipage}{0.69\textwidth}
\begin{tabular}{cc}
\textit{Configuration}& 
\textit{Layered Word}\\\\
${f(h(g(x)+1))}$
& $\langle h, s_0\rangle \langle f, s_0\rangle\$ \langle g, s_1\rangle$\\$\downarrow $ & $\downarrow$\\
${f(g(g(x))+1)}$ & \multicolumn{1}{l}{$\langle f, s_0\rangle\$ \langle g, s_1\rangle {\langle g, s_1\rangle} =\langle f, s_0\rangle\$ \App^{s_1}[g](\langle g, s_1\rangle)$}
\end{tabular}
\end{minipage}
\begin{minipage}{0.24\textwidth}\footnotesize
$$\begin{array}{|l|}\hline
{\mathrm{Program}} \\\hline\\
{f(0)=0;}\\ {f(x+1)=}\\
\qquad{f(g(x+1))+1;}\\
\\\\
{g(0)=0;}\\ {g(x+1)=h(x);}\\\\
{h(0)=0;}\\ {h(x+1)=g(x)+1;}\\\hline
\end{array}$$
\end{minipage}
\caption{Sample model of the call-stack restructuring by using the function ${\App}^{s_j}$}
\label{fig:AppProg}
\end{figure}

\item Insert function $\Lift^{s_j}[\Psi\langle s_k\rangle]$ models adding a ready function call to the invisible layer. The ready function call is a call whose pattern can be matched immediately. E.\,g., as shown in Figure~\ref{fig:InsProg} (the program rules are the same as for $\App^{s_j}$).

\begin{figure}[h!]
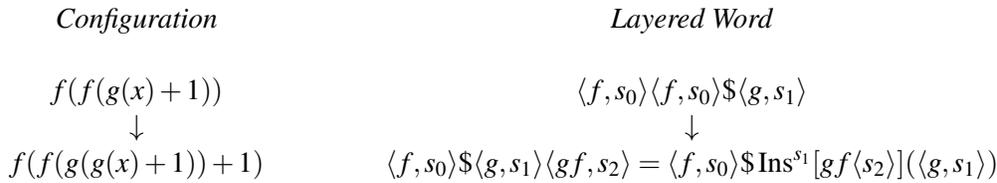

\begin{tabular}{ccc}
\textit{Configuration}& &
\textit{Layered Word}\\\\

${f(f(g(x)+1))}$
& &$\langle f, s_0\rangle \langle f, s_0\rangle\$ \langle g, s_1\rangle$\\
$\downarrow $ & &$\downarrow$\\
${f(f(g(g(x)+1))+1)}$ & \qquad\qquad &\multicolumn{1}{l}{  $\langle f, s_0\rangle\$ \langle g, s_1\rangle {\langle gf, s_2\rangle}=\langle f, s_0\rangle\$ \Lift^{s_1}[gf\langle s_2\rangle](\langle g, s_1\rangle)$}\end{tabular}

\caption{Sample model of the call-stack restructuring by using the function $\Lift^{s_j}[\Psi\langle s_k\rangle]$}
\label{fig:InsProg}
\end{figure}

\item Deleting $\Del^{s_j}$ corresponds to replacement of one argument of a function call by an expression without function calls. E.\,g., as shown in Figure~\ref{fig:DelProg}.

\begin{figure}[ht]
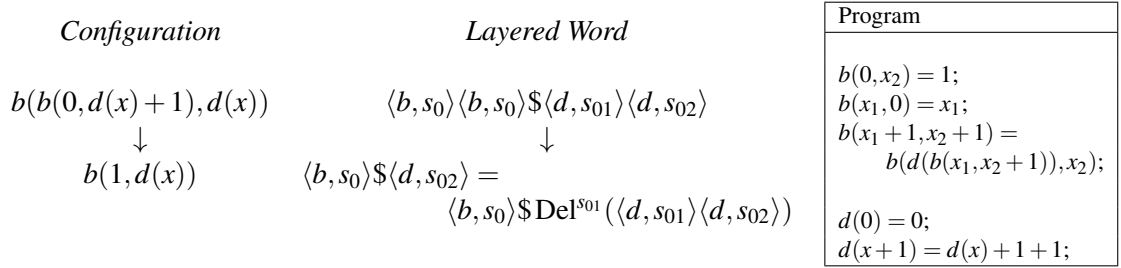

\begin{minipage}{0.68\textwidth}
\begin{tabular}{cc}
\textit{Configuration}& 
\textit{Layered Word}\\\\
${b(b(0, d(x)+1),d(x))}$
& $\langle b, s_0\rangle \langle b, s_0\rangle\${\langle d, s_{01}\rangle} \langle d, s_{02}\rangle$
\\$\downarrow $ & $\downarrow$\\
${b(1,d(x))}$ &  \multicolumn{1}{l}{$\langle b, s_0\rangle\$ \langle d, s_{02}\rangle = $}
\\
&\multicolumn{1}{r}{\qquad\qquad\quad $\langle b, s_0\rangle\$ \Del^{s_{01}} (\langle d, s_{01}\rangle \langle d, s_{02}\rangle)$}
\end{tabular}
\end{minipage}
\begin{minipage}{0.25\textwidth}\footnotesize
$$\begin{array}{|l|}
\hline
{\mathrm{Program}} \\\hline\\
{b(0,x_2)=1;}\\ 
{b(x_1,0)=x_1;}\\ 
{b(x_1+1,x_2+1)=}\\
\qquad{b(d(b(x_1,x_2+1)),x_2);}
 \\\\
{d(0)=0;}\\
{d(x+1)=d(x)+1+1;}
\\\hline
\end{array}$$
\end{minipage}

\caption{Sample model of the call-stack restructuring by using the function $\Del^{s_j}$}
\label{fig:DelProg}
\end{figure}

\item Copying $\Copy^{s_j}$ corresponds to copying one argument of a function call into another. E.\,g., as shown in Figure~\ref{fig:CopyProg} (the program rules are the same as for $\Del^{s_j}$). Since the semantics is call-by-name, the transformation does not retain information about equality of the call and its copy.

\begin{figure}[ht]
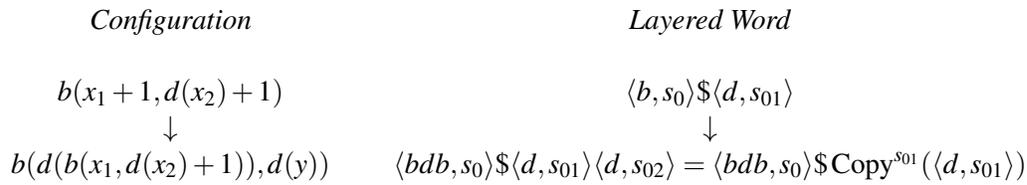

\begin{tabular}{ccc}
\textit{Configuration}& &
\textit{Layered Word}\\\\
${b(x_1+1,d(x_2)+1)}$
& &{$\langle b, s_0\rangle\$ \langle d, s_{01}\rangle$}
\\$\downarrow $ & & $\downarrow$\\
${b(d(b(x_1,d(x_2)+1)),d(y))}$ &\quad&  \multicolumn{1}{l}{$\langle bdb, s_0\rangle\$ \langle d, s_{01}\rangle {\langle d, s_{02}\rangle} = \langle bdb, s_0\rangle\$ \Copy^{s_{01}} (\langle d, s_{01}\rangle)$}
\end{tabular}
\caption{Sample model of the call-stack restructuring by using the function $\Copy^{s_j}$}
\label{fig:CopyProg}
\end{figure}

\end{enumerate}

\vspace{10pt}
Because every function definition is finite, only a finite number of append, delete, copy and insert functions can be applied in one step of the unfolding. That guarantees finiteness of the constants $K_1$ and $K_2$ in the corresponding multi-layer grammars. 

\section{Turchin's Relation and Multi-Layer Grammars}
\label{section:TurRelMLPG}

\begin{Definition}
Let $\mathbf{G}$ be a multi-layer prefix grammar with the set of rules $\mathbf{R}$ such that every rule from $\mathbf{R}$ rewrites at most $N$ letters of a visible layer. Given a trace $\{\Gamma_k\$\Delta_k\}$ and its segment $[i,j]$, suffix $\Theta$ of $\Phi_i$ is called \emph{a~permanently stable} suffix w.r.t. the segment $[i,j]$ if all the words $\Gamma_k\$\Delta_k$, $i\leq k<j$, are of the form $\Phi_k\Theta\$\Delta_k$ where $\Phi_k$ is a layered word with the length not less than $N$, and $\Gamma_j$ is of the form $\Phi_j\Theta$, where $\Phi_j$ may be $\Lambda$ \footnote{In the case of alphabetic prefix grammars, when $N=1$, the first condition implies the second.}. If $j$ is not bounded, $\Theta$ is called \emph{a permanently stable suffix} w.r.t. $i$.
\end{Definition}

Informally, a permanently stable suffix is a suffix of the visible layer that is never changed in the trace segment $[i,j]$. In terms of call stack behavior, a permanently stable suffix corresponds to an unchanged context of the computation.

\begin{Example}
Let a trace of some alphabetic multi-layer grammar be:

\UseAllTwocells
\[
\xymatrix @-3mm {
*+{\Phi_0: \langle hf, s_0\rangle \$ \langle h, s_1\rangle}\ar[r] & 
*+{\Phi_1: \langle f, s_0\rangle \$ \langle hg, s_1\rangle}\ar[r] & 
*+{\Phi_2: \langle gf, s_0\rangle \$ \langle hg, s_1\rangle}\ar[r] & *+{\Phi_3: \langle hg, s_1\rangle \langle gf, s_0\rangle \$\Lambda} 
}
\] 
Suffix $\langle f, s_0\rangle$ of $\Phi_i$ is permanently stable w.r.t. $[2,3]$, but is not permanently stable w.r.t. the position $0$ (i.\,e., w.r.t. the whole trace) because in the word $\Phi_1$ it is preceded by the empty prefix in the visible layer\footnote{Thus, the letter $f$ in $\Phi_1$ is rewritten by some rule to $gf$ in $\Phi_2$.}. Suffix $\langle g, s_1\rangle$ of $\Phi_1$ is not permanently stable w.r.t. $[1,3]$, because in the words $\Phi_1$, $\Phi_2$ it occurs in the invisible layer.
\end{Example}

\begin{Definition}
Let $\mathbf{G}=\langle \Upsilon, \mathbf{S},R, \mathfrak{F}^v_{K_1, K_2},\Gamma_0\$\Delta_0\rangle$ be a multi-layer prefix grammar. Given two~$\mathbf{G}$-words $\Xi_i = \Gamma_i\$\Delta_i$, $\Xi_j= \Gamma_j\$\Delta_j$ in a trace $\{\Gamma_k\$\Delta_k\}$, we say that the words form \emph{a Turchin pair}  (denoted as $\Xi_i\preceq \Xi_j$) if $\Gamma_i = \Phi\Theta_0$, $\Gamma_j = \Phi'\Psi\Theta_0$, $\Phi$ is equal to $\Phi'$ as a plain word (up to the layer labels) and the suffix $\Theta_0$ is permanently stable w.r.t. segment $[i,j]$.
\end{Definition}

Thus, if we do not take into account the invisible layer and layer labels, then the definition of the~Turchin pair for the traces generated by the multi-layer prefix grammars repeats the definition for the plain prefix grammars given in Section~\ref{Section:Preliminaries}.

\begin{Theorem}[Strengthened Turchin's Theorem]
\label{prop:infinite}
Let $\mathbf{G}=\langle \Upsilon, \mathbf{S},\mathbf{R}, \mathfrak{F}^v_{K_1, K_2},\Gamma_0\$\Delta_0\rangle$ be a multi-layer prefix grammar. Every infinite trace generated by grammar $\mathbf{G}$ contains an infinite subsequence $\{\Gamma_{k}\$\Delta_{k}\}$ such that for every $\Gamma_{{k_1}}\$\Delta_{{k_1}}$, $\Gamma_{{k_2}}\$\Delta_{{k_2}}$, $k_1<k_2$ implies  $\Gamma_{{k_1}}\preceq\Gamma_{{k_2}}$.
\end{Theorem}

\begin{proof}

The idea of the proof is borrowed from the original V.~Turchin's work \cite{Turchin88}. 

Let $N$ be the maximal number of the letters in the visible layer that can be changed by rewriting rules from $\mathbf{R}$. We consider the following two cases.

Let $\{\Gamma_i\$\Delta_i\}$ be an infinite trace under the theorem conditions. If some word $\Gamma_i\$\Delta_i$ contains a suffix $\Theta_i$ that is permanently stable w.r.t. $i$ in the trace, and for all $j>i$ no word $\Gamma_j\$\Delta_j$ contains a permanently stable (w.r.t. $j$) suffix being longer than $\Theta_i$, then there are infinitely many words $\Phi\$\Psi$ in the trace such that $|\Phi|\leq|\Theta_i|+N$. Thus some word $\Phi$ repeats itself as a plain word in the visible parts infinitely many times. The infinite subsequence of the words having $\Phi$ as the plain word in the visible part is a subsequence of $\{\Gamma_i\$\Delta_i\}$, such that every two words of the subsequence form a Turchin pair.

Let $\{\Gamma_i\$\Delta_i\}$ be an infinite trace with no upper bound on the~permanently stable suffixes' length. So there is an infinite sequence of words $\{\Gamma_{i_n}\$\Delta_{i_n}\}$ such that the visible part $\Gamma_{i_n}$ contains the suffix $\widehat{\Phi}_{i_n}$ that is permanently stable w.r.t. $i_n$, but not w.r.t. $i_n-1$. Namely, such words $\Gamma_{i_n}\$\Delta_{i_n}$ are the words where the first letter of $\widehat{\Phi}_{i_n}$ is generated in the visible part. The letter of $\Gamma_{i_n}$ preceding $\widehat{\Phi}_{i_n}$ does not belong to the permanently stable suffix, so it is erased somewhere further in the trace. The $\mathbf{G}$-word in which it is erased looks as $\Psi_{i_n}\widehat{\Phi}_{i_n}\$\Delta_{i_{k}}$, $|\Psi_{i_n}|\leq N$. 

Since $|\Psi_{i_n}|$ is bounded, there exists at least one plain word $\Psi$ such that the sequence $\{\Psi_{i_n}\widehat{\Phi}_{i_n}\$\Delta_{i_k}\}$ contains infinitely many $\mathbf{G}$-words with the prefix $\Psi_{i_n}$ equal to $\Psi$ as a plain word. Let the subsequence of such words be $\{\Gamma_{i_l}\$\Delta_{i_l}\}$. For every ${i_{l_1}}<{i_{l_2}}$, $\Gamma_{i{l_1}}=\Psi_{i{l_1}}\widehat{\Phi}_{i_{l_1}}$, $\Gamma_{i_{l_2}}=\Psi_{i_{l_2}}\widehat{\Phi}_{i_{l_2}}=\Psi_{i_{l_2}}\Xi\widehat{\Phi}_{i_{l_1}}$, the suffix $\widehat{\Phi}_{i_{l_1}}$ is unchanged in $\Gamma_{i_{l_2}}$, and $\Psi_{i_{l_1}}$ and $\Psi_{i_{l_2}}$ coincide as the plain words. Thus, the sequence $\{\Gamma_{i_l}\$\Delta_{i_l}\}$ satisfies the statement of the theorem.
\end{proof}

Theorem \ref{prop:infinite} implies the following corollary.

\begin{Prop}
A composition of Turchin's relation and an arbitrary well binary relation $R$ is a well binary relation with respect to the set of the traces generated by the multi-layer prefix grammars.
\end{Prop}

\begin{proof}
Every infinite trace generated by a multi-layer prefix grammar contains an infinite subsequence, every two words of which form a Turchin pair. Due to the well-binariness of relation $R$, this subsequence also contains two words $\Gamma$ and $\Delta$ such that $\Gamma$ precedes $\Delta$ and $(\Gamma,\Delta)\in R$.
\end{proof}

\begin{Example}
\label{Example:log2progcomposite}

We turn back to the program of Example~\ref{Example:log2prog}. Example~\ref{Example:log2proggen} shows that generalization w.r.t. Turchin's relation produces a better residual program than generalization w.r.t. the homeomorphic embedding. Let us consider the composition of these two relations. The result of the unfolding and the generalization is shown in Figure~\ref{fig:treegenTurHom} given in Appendix.

The residual program extracted from the graph after the generalization looks as follows.

$$\begin{array}{|l|l|}\hline \textrm{Residual program generated by} \\ 
\textrm{the graph of Figure~\ref{fig:treegenTurHom}}\\\hline
{f_1(0)=0;}\\ {f_1(1)=1;}\\
{f_1(2)=1;}\\ {f_1(x+1+1+1)=f_1(g_1(x)+1);}\\\\
{g_1(0)=0;}\\ {g_1(1)=0;}\\
{g_1(x+1+1)=g_1(x)+1;}\\\hline
\end{array}$$

\noindent This residual program is more efficient than the two residual programs given in Example~\ref{Example:log2proggen}, and contains less rules.

Why the composition helps to construct a good generalization in this case? Let us consider the configurations in Turchin's relation along the trace. The first two such configurations are $f(g(g(x_1)+1))+1$ and $f(h(g(x_1)))+1$. The innermost call of the first configuration does not appear the call-stack, so from the point of view of Turchin's relation these two configurations look as $f(g(z_1))$ and $f(h(g(z_2)))$ (and the call of $f$ is not changed in the corresponding segment of the trace). So, if one uses only Turchin's relation, the two configurations are to be generalized. The homeomorphic embedding relation does not ignore the passive parts of the configurations. But if we use only the homeomorphic embedding relation, we will meet problems as described in Example~\ref{Example:log2proggen}. Thus, using the composition allows a supercompiler to generalize configurations $\theta_1$ and $\theta_2$ only if both the call stacks of $\theta_1$ and $\theta_2$ are similar and the static data in $\theta_1$ and $\theta_2$ are also similar, which results in more accurate generalizations.
\end{Example}

\section{Conclusion}

Turchin's relation for call-by-name computations is a strong and consistent branch termination criterion. Every infinite trace generated by unfolding a program in the language with the call-by-name semantics contains two elements, whose call-stack configurations form a Turchin pair. Neither the homeomorphic embedding can replace Turchin's relation nor Turchin's relation can be considered as a simplification of the homeomorphic embedding in the case of the normal-order reduction\footnote{If all the configurations in the computation tree contained no branching (i.e., were consisting only of unary function calls and constructors), Turchin's relation could be considered as a ``one-gap version'' of the homeomorphic embedding for the call-stacks, and could be replaced by an annotated version of the homeomorphic embedding as shown in \cite{ANNTur}. However, for terms having the tree structure (even with only one branching node), the ``one-gap'' (or even ``n-gap'') relation is not well-binary (that is also shown in \cite{ANNTur}). So, considering the call-stack configurations as words is somewhat essential for making Turchin's relation well-binary.}. Turchin's relation can be used together with the homeomorphic embedding relation without the loss of well-binariness, and can be used not only for deciding when to terminate the computation path but also for deciding how to generalize the configurations. 

The alphabetic multi-layer grammars, describing the function call stack behaviour, are able to generate languages with very long words, but it seems to us they are not able to generate languages with words having a complex structure. It would be interesting to find some practical problems, which can be solved with the power of Turchin's relation (or homeomorphic embedding relation) on the call-stack configurations for call-by-name computations.  

\section*{Acknowledgements}
I would like to thank A.\,P.~Nemytykh for many fruitful advices and help in improving the paper, and the anonymous referees for the useful feedback.

%

\label{sect:bib}

\nocite{*}
\bibliographystyle{eptcs}
\bibliography{vpt2016}

\vspace{+15pt}
\section*{Appendix}
\label{section:Appendix}

\vspace{+10pt}
\subsection{Proof of Theorem \ref{Theorem:CompComp} (Subsection~\ref{subsection:forMLPG})}
\label{proof:thcompcomp}

\begin{proof}

It is sufficient to prove that, given an input word, every Turing machine on the input can be emulated by a multi-layer prefix grammar treating the input word as its initial word. 

\vspace{+5pt}
Consider an arbitrary Turing machine $\langle Q, \Upsilon_A, b, \sigma, q_0, F\rangle$, where $Q$ is a finite state alphabet, $\Upsilon_A$ is a finite tape alphabet, $b$ is the blank symbol, $q_0\in Q$ is the initial state, $F\subset Q$ is a set of the final states, $\sigma\subset Q\times \Upsilon_A\times Q\times\Upsilon_A\times\{L, R\}$ is a set of transition rules, and an input $I\in \Upsilon_A^*$. Let us introduce a~multi-layer prefix grammar with the alphabet $\Upsilon=\Upsilon_A\cup Q^R\cup Q^L\cup\{Blank^L, Blank^R, b\}$, where $Q^R$ and $Q^L$ are the state alphabet $Q$ marked by the superscripts meaning ``a state after moving to the right cell'' and ``a state after moving to the left cell'' correspondingly. $Blank^R$ and $Blank^L$ are special ``end-marks'' referring to the blanks on the tape after the rightmost and leftmost cells reached by the machine head in the~computation. All blanks on the tape between them are denoted in the model grammar by usual $b$ symbols. 

\vspace{+5pt}
Let us assume in this proof that all labels in the model grammar are of the form $s_i$, where $i$ is a rational number ($i\in\mathbb{Q}$), and $s_i\triangleleft s_j$ iff $i<j$ ($i,j\in \mathbb{Q}$).

\vspace{+5pt}
The initial word in the model grammar is

\[\Gamma_0\$\Delta_0=\langle q_0^R, s_0 \rangle I\langle s_0 \rangle \langle Blank^R, s_0 \rangle\$\langle Blank^L, s_1 \rangle\]

In order to emulate a rule $(q_1, a_1)\rightarrow (q_2, a_2, R)\in \sigma$, we use one of the following two rewrite schemes, where $x$ is a letter variable:

\[\langle q_1^R, s_i\rangle \langle a_1, s_j\rangle \Phi\$ \Psi \langle x, s_k\rangle\rightarrow \langle q_2^R, s_j\rangle \Phi\$ \Psi \langle x, s_k\rangle \langle a_2, s_{\frac{j+k}{2}}\rangle\]  

\[\langle a_1, s_j\rangle \langle q_1^L, s_i\rangle  \Phi\$ \Psi \langle x, s_k\rangle\rightarrow \langle q_2^R, s_i\rangle \Phi\$ \Psi \langle x, s_k\rangle \langle a_2, s_{\frac{i+k}{2}}\rangle\] 

\vspace{+5pt}
In order to emulate a rule $(q_1, a_1)\rightarrow (q_2, a_2, L)$, we use one of the following two schemes ($s_k=\child(s_i)$):

\[\langle q_1^R, s_i\rangle \langle a_1, s_j\rangle \Phi\$ \Psi \langle x, s_k\rangle \rightarrow \langle x, s_k\rangle \langle q_2^L, s_j\rangle \langle a_2, s_j\rangle\Phi\$ \Psi\]

\[\langle a_1, s_j\rangle \langle q_1^L, s_i\rangle  \Phi\$ \Psi \langle x, s_k\rangle\rightarrow  \langle x, s_k\rangle \langle q_2^L, s_i\rangle \langle a_2, s_i\rangle\Phi\$ \Psi\]

If $q_2\in F$, we put the corresponding rule scheme into the set of the halting rules.
\vspace{10pt}

So we model the tape part to the right of the machine head by the visible layer, and the tape part to the left of the machine head by the invisible layer. The letters from $Q^R\cup Q^L$ can only appear in the visible layer, and what is more, any~letter from $Q^R$ may be only the first letter of the visible layer, and any letter from $Q^L$ --- the~second letter of the visible layer.
\end{proof}

\vspace{+10pt}
\subsection{Proof of Proposition \ref{Prop:CompAlphabeticMLPG} (Subsection~\ref{subsection:forMLPG})}
\label{proof:compalph}

\begin{proof}
First, we prove that the described class of the multi-layer prefix grammars is not weaker than the class of the context-free grammars and the class of the one-state Turing machines.

\vspace{+5pt}
For the one-state Turing machines, the corresponding model is constructed in the proof of Theorem~\ref{Theorem:CompComp}. 

\vspace{+5pt}
Given a context-free language, we consider its generating context-free grammar $\mathbf{C}$ in the Greibach normal form \cite{Greibach}. Let the set of non-terminals of $\mathbf{C}$ be $Q$ and the set of terminals be $T$. We construct a multi-layer grammar with the alphabet $\Upsilon=T\cup Q\cup\{\pop\}$. The initial word is $\langle S, s_0 \rangle \langle \pop, s_0\rangle \$\Lambda$ where $S$ is the initial symbol of $\mathbf{C}$. 

\vspace{+5pt}
For every rule $q_1\rightarrow u q_2 q_3$, $q_i\in Q$, $u\in T$ of the context-free grammar $\mathbf{C}$, we construct the~rewriting rule 
\[\langle q_1, s_0\rangle\Phi\$\Psi\rightarrow \langle q_2 q_3, s_0\rangle\Phi\$\Psi \langle u, s_1\rangle\] 

\vspace{+5pt}
For a rule $q_1\rightarrow \Lambda$ we construct the rule \[\langle q_1, s_0\rangle \Phi\$\Psi\rightarrow \Phi\$\Psi\] Finally, we add the halting rule $R_{\nil}$ \[\langle\pop, s_0\rangle\Phi\$\Psi \rightarrow \Psi\langle s_1\rangle\Phi\$\Psi\langle\ovl{$s_1$}\rangle\]

Because all invisible letters are labelled by $s_1$ and $\pop$ can only be generated in the initial word, rule $R_{\nil}$ actually looks as $\langle \pop, s_0\rangle\$\Psi\langle s_1 \rangle \rightarrow \Psi\langle s_1\rangle\$\Lambda$. 

\vspace{+5pt}
Since the alphabetic multi-layer grammars generate all context-free languages and all languages generated by one-state Turing machines, they are stronger than the one-state Turing machines (some regular languages cannot be generated by such a machine \cite{OneStateTM}).

\vspace{+10pt}
In order to prove that the alphabetic multi-layer prefix grammars are stronger than the context-free grammars, it is sufficient to show that the language $\{b^{2^n}|n\in\mathbb{N}\}$ can be generated by such a~grammar (such a language cannot be generated by any context-free grammar \cite{HopcroftUlmanBook}). We assume that $s_{\Gamma}\triangleleft s_{\Delta}$ iff the word $\Gamma$ is a prefix of  the word $\Delta$. Otherwise, $s_{\Gamma}$ and $s_{\Delta}$ are incomparable.

\vspace{+5pt}
Let the initial word of the grammar $\mathbf{G}_{\textrm{ExpLang}}$ be $\langle a, s_0\rangle \$ \langle bb, s_{01}\rangle$. The set of rewriting rules of $\mathbf{G}_{\textrm{ExpLang}}$ is as follows.

\vspace{+10pt}
\begin{tabular}{lll}
$R^{[1]}: \langle a, s_0\rangle \Phi\$\Psi\rightarrow \Psi\langle \child(s_0)\rangle\langle a, s_0\rangle\Phi\$\Psi\langle \ovl{$\child(s_0)$}\rangle $ \\
$R^{[2]}: \langle a, s_0\rangle\$\Psi \rightarrow \Psi\langle \child(s_0)\rangle \$\Psi\langle \ovl{$\child(s_0)$}\rangle $ \\
$R^{[3]}: \langle b, s_i\rangle \Phi\$\Psi \rightarrow \Phi\$\App^{s_{i1}}[bb](\Psi)$ \\
\end{tabular}

\vspace{+5pt}
\noindent Rule $R^{[2]}$ is the halting rule. 

\vspace{+15pt}
We consider the tree of the possible traces generated by grammar $\mathbf{G}_{\textrm{ExpLang}}$ (its fragment is shown in Figure~\ref{figure::explanguage}). All the letters in the invisible parts of the words in the trace have a single label. Thus, when a pop rule is applied to a word in the trace, it pops all the letters from the invisible layer.

\vspace{+10pt}
\begin{figure}[h!]
\label{figure::explanguage}

\UseAllTwocells
\[
\xymatrix @-4mm {
&*+{\langle a, s_0\rangle \$ \langle bb, s_{01}\rangle}
\ar[dr]^{R^{[1]}}\ar[dl]_{R^{[2]}}
\\
*+{\langle bb, s_{01}\rangle \$\Lambda} && *+{\langle bb, s_{01}\rangle \langle a, s_0\rangle \$ \Lambda}\ar[d]^{R^{[3]}}
\\
&&*+{\langle b, s_{01}\rangle \langle a, s_0\rangle \$ \langle bb, s_{011}\rangle}\ar[d]^{R^{[3]}}
\\
&&*+{\langle a, s_0\rangle \$ \langle bbbb, s_{011}\rangle}\ar[dl]_{R^{[2]}}\ar[d]^{R^{[1]}}
\\
&*+{\langle bbbb, s_{011}\rangle \$\Lambda}&  {\dots}\ar[d]
\\
&&*+{\langle a, s_0\rangle \$ \langle b^8, s_{0111}\rangle }\ar[dr]^<(0.4){R^{[1]}}\ar[dl]_{R^{[2]}}
\\
&{\dots}&&{\dots}
}
\] 

\caption{A fragment of the tree of the possible traces generated by $\mathbf{G}_{\textrm{ExpLang}}$}
\end{figure}
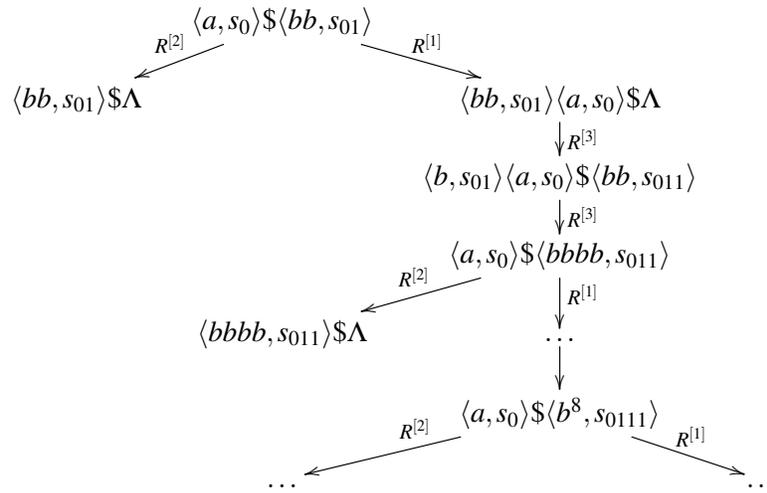

\vspace*{+25pt}
The resulting language $\{b^{2^n}|n\in\mathbb{N}\}$ is not even mildly context-sensitive\footnote{The class of mildly context-sensitive grammars is a special subclass of context-sensitive grammars that includes not only all context-free grammars, but also, e.\,g., tree adjoining grammars \cite{Joshi}.} \cite{HopcroftUlmanBook}. 
\end{proof}

\newpage
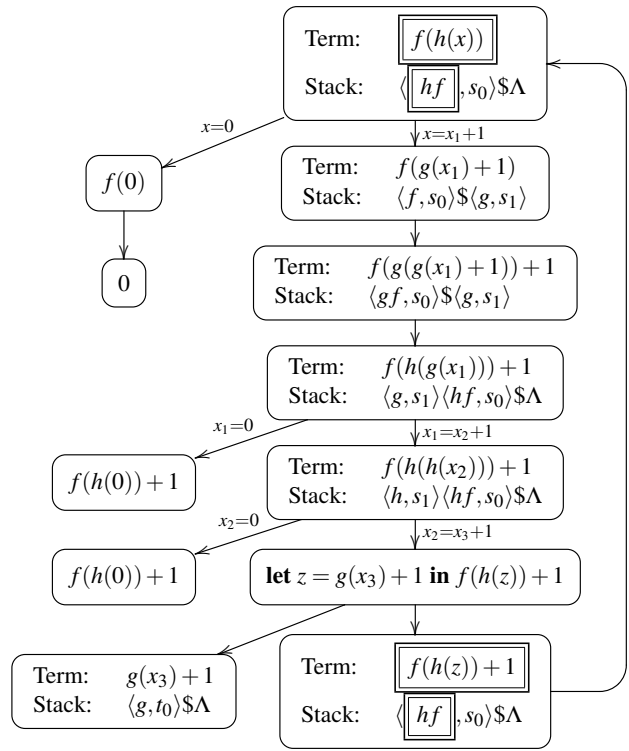
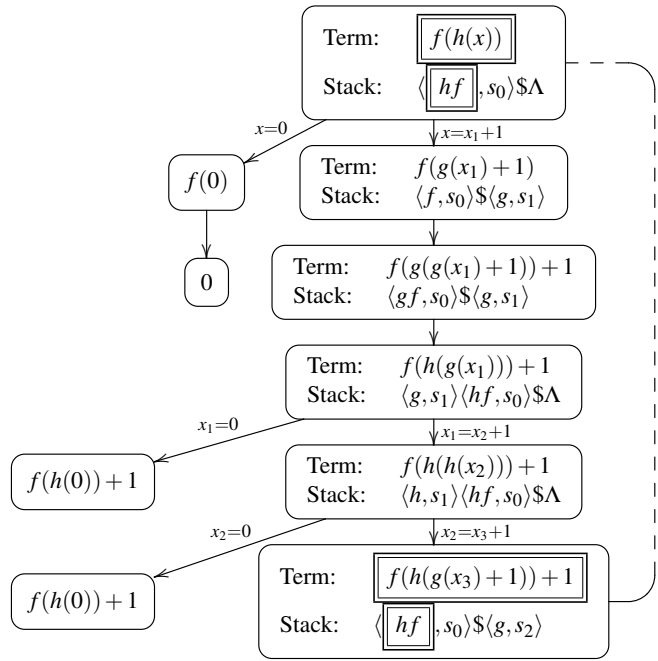
\begin{figure}[t]
\begin{center}
\footnotesize
\begin{tabular}{l}
 $$
\xymatrix @-5mm @C=0.6em{
&&*+[F-:<5pt>]{\begin{array}{ll} \textrm{Term: }&\doublebox{$f(h(x))$}\\\textrm{Stack: }&{\langle \doublebox{$hf$}, s_0\rangle \$\Lambda}\end{array}}\ar[dl]_<(0.35){x=0}\ar[d]^<(0.22){x=x_1+1}\\
&*++[F-:<5pt>]{f(0)}\ar[d] & *+[F-:<5pt>]{\begin{array}{ll} \textrm{Term: }&f(g(x_1)+1)\\\textrm{Stack: }&{\langle f, s_0\rangle \$ \langle g, s_1\rangle}\end{array}}\ar[d]\\
&*++[F-:<5pt>]{0}  & *+[F-:<5pt>]{\begin{array}{ll} \textrm{Term: }&f(g(g(x_1)+1))+1\\\textrm{Stack: }&{\langle gf, s_0\rangle \$ \langle g, s_1\rangle} \end{array}}\ar[d]\\
&& *+[F-:<5pt>]{\begin{array}{ll} \textrm{Term: }&f(h(g(x_1)))+1\\\textrm{Stack: }&{\langle g, s_1\rangle\langle hf, s_0\rangle \$ \Lambda}\end{array}}\ar[dll]_<(0.35){x_1=0}\ar[d]^{x_1=x_2+1}\\
*++[F-:<5pt>]{f(h(0))+1}& & *+[F-:<5pt>]{\begin{array}{ll} \textrm{Term: }&f(h(h(x_2)))+1\\\textrm{Stack: }&{\langle h, s_1\rangle\langle hf, s_0\rangle \$ \Lambda}\end{array}}\ar[dll]_<(0.35){x_2=0}\ar[d]^>(0.8){x_2=x_3+1}&*++++{}\\
*++[F-:<5pt>]{f(h(0))+1} && *+[F-:<5pt>]{\begin{array}{ll} \textrm{Term: }&\doublebox{$f(h(g(x_3)+1))+1$}\\\textrm{Stack: }&{\langle \doublebox{$hf$}, s_0\rangle \$ \langle g, s_2\rangle}\end{array}}\ar@{--}`r[ru]`[uuuuul][uuuuu] 
}
$$ \\ \\
\normalsize Generalization of ${f(h(g(x_3)+1))+1}$ and ${f(h(x))}$ is built
\end{tabular}

\vspace{7.2pt}
\begin{tabular}{l}
 $$
\xymatrix @-5mm @C=0.6em{
&*+[F-:<5pt>]{\begin{array}{ll} \textrm{Term: }&\doublebox{$f(h(x))$}\\\textrm{Stack: }&{\langle \doublebox{$hf$}, s_0\rangle \$\Lambda}\end{array}}\ar[dl]_<(0.35){x=0}\ar[d]^<(0.22){x=x_1+1}\\
*++[F-:<5pt>]{f(0)}\ar[d] & *+[F-:<5pt>]{\begin{array}{ll} \textrm{Term: }&f(g(x_1)+1)\\\textrm{Stack: }&{\langle f, s_0\rangle \$ \langle g, s_1\rangle}\end{array}}\ar[d]\\
*++[F-:<5pt>]{0}  & *+[F-:<5pt>]{\begin{array}{ll} \textrm{Term: }&f(g(g(x_1)+1))+1\\\textrm{Stack: }&{\langle gf, s_0\rangle \$ \langle g, s_1\rangle} \end{array}}\ar[d]\\
 & *+[F-:<5pt>]{\begin{array}{ll} \textrm{Term: }&f(h(g(x_1)))+1\\\textrm{Stack: }&{\langle g, s_1\rangle\langle hf, s_0\rangle \$ \Lambda}\end{array}}\ar[dl]_<(0.35){x_1=0}\ar[d]^{x_1=x_2+1}\\
*++[F-:<5pt>]{f(h(0))+1} & *+[F-:<5pt>]{\begin{array}{ll} \textrm{Term: }&f(h(h(x_2)))+1\\\textrm{Stack: }&{\langle h, s_1\rangle\langle hf, s_0\rangle \$ \Lambda}\end{array}}\ar[dl]_<(0.3){x_2=0}\ar[d]^<(0.22){x_2=x_3+1}\\
*++[F-:<5pt>]{f(h(0))+1} & *++[F-:<5pt>]{\mathbf{let}\textrm{ }{z=g(x_3)+1}\textrm{ }\mathbf{in}\textrm{ }{f(h(z))+1}}\ar[dl]\ar[d]&*++++{}\\
*+[F-:<5pt>]{\begin{array}{ll} \textrm{Term: }&g(x_3)+1\\\textrm{Stack: }&{\langle g, t_0\rangle \$ \Lambda}\end{array}}&*+[F-:<5pt>]{\begin{array}{ll} \textrm{Term: }&\doublebox{$f(h(z))+1$}\\\textrm{Stack: }&{\langle \doublebox{$hf$}, s_0\rangle \$ \Lambda}\end{array}}\ar@{->}`r[ru]`[uuuuuul][uuuuuu]
}
$$ \\ \\
\normalsize A fragment of the graph after the generalization
\end{tabular}
\end{center}
\caption{Generalization using the composition of Turchin's relation and the homeomorphic embedding. Example~\ref{Example:log2progcomposite}}
\label{fig:treegenTurHom}
\end{figure}
\end{document}